\documentclass{lmcs}
\pdfoutput=1

\usepackage{lastpage}
\lmcsdoi{15}{4}{7}
\lmcsheading{}{\pageref{LastPage}}{}{}%
{Mar.~02,~2018}{Oct.~30,~2019}{}

\usepackage{etex}      
\usepackage{amsmath}		
\usepackage{amssymb}		
\usepackage{cmll}
\usepackage{stmaryrd}
\usepackage{proof}
\usepackage{listings}
\usepackage[all]{xy}

\usepackage{graphicx}		
\usepackage{wrapfig}
\usepackage{color}		
\usepackage{hyperref}
\usepackage{tikz}
\usetikzlibrary{calc, intersections, shapes.geometric}
\tikzset{every picture/.style={
		scale=0.15,
		font={\fontsize{8pt}{12}\selectfont},
		baseline={(current bounding box.west)}
}} 

\newcommand{\abs}[3]{\lambda {#1}^{#2}. #3}
\newcommand{\app}[2]{#1 \, #2}	      

\newcommand{\FV}{\mathit{FV}}
\newcommand{\plug}[2]{#1 \langle #2 \rangle}
\newcommand{\emptyCtxt}{\langle \cdot \rangle}

\newcommand{\id}[1]{\abs{#1}{}{#1}}

\newcommand{\appGen}[2]{#1 \mathbin{\$} #2}
\newcommand{\appLazy}[2]{#1 \mathbin{@} #2}
\newcommand{\appLR}[2]{#1 \mathbin{\overrightarrow{@}} #2}
\newcommand{\appRL}[2]{#1 \mathbin{\overleftarrow{@}} #2}
\newcommand{\W}[1]{\llparenthesis #1 \rrparenthesis}

\renewcommand{\oc}{\mathop{!}}

\newcommand{\twoheadmapsto}{
\mathrel{\ooalign{$\twoheadrightarrow$\cr%
\kern-.15ex\raisebox{.2ex}{\scalebox{1}[0.8]{$\shortmid$}}\cr}}}
\newcommand{\longtwoheadmapsto}{
\mathrel{\ooalign{$\longtwoheadrightarrow$\cr%
\kern-.15ex\raisebox{.2ex}{\scalebox{1}[0.8]{$\shortmid$}}\cr}}}

\newcommand{\up}{\mathord{\uparrow}}
\newcommand{\dn}{\mathord{\downarrow}}

\newcommand{\lb}[1]{\mathit{#1}}

\newcommand{\Init}{\mathit{Init}}
\newcommand{\Final}{\mathit{Final}}

\newcommand{\ev}{\mathit{Eval}}
\newcommand{\ex}{\mathit{Exec}}


\begin{document}

\title[The Dynamic Geometry of Interaction Machine]{
The Dynamic Geometry of Interaction Machine: \\
A Token-Guided Graph Rewriter{\rsuper*}}
\titlecomment{
{\lsuper*}This paper expands the work presented at WPTE
2017~\cite{MuroyaG17CBV}, which extends the initial presentation at
CSL 2017~\cite{MuroyaG17}.
}

\author[K.~Muroya]{Koko Muroya}	
\address{School of Computer Science, University of Birmingham, UK}	
\email{\{k.muroya,d.r.ghica\}@cs.bham.ac.uk}  

\author[D.\,R.~Ghica]{Dan R. Ghica}	





\begin{abstract}
 In implementing evaluation strategies of the lambda-calculus, both correctness and 
 efficiency of implementation are valid concerns.
 While the notion of correctness is determined by the evaluation
 strategy, regarding efficiency there is a larger design space that can be explored, in
 particular the trade-off between space  versus time efficiency.
 Aiming at a unified framework that would enable the study of this
 trade-off, we introduce
 an abstract machine, inspired by 
 Girard's Geometry of Interaction (GoI), a machine combining token
 passing and graph rewriting. 
 We show soundness and completeness of our abstract machine, called
 the \emph{Dynamic GoI Machine} (DGoIM), with
 respect to three evaluations: call-by-need, left-to-right
 call-by-value, and right-to-left call-by-value.
 Analysing time cost of its execution classifies the machine as
 ``efficient'' in Accattoli's taxonomy of abstract machines.
\end{abstract}


\maketitle


\section{Introduction}

The lambda-calculus is a simple yet rich model of computation, relying on 
a single mechanism to activate a function in
computation, beta-reduction, that replaces function parameters with
actual input.
While in the lambda-calculus itself 
beta-reduction can be applied in an unrestricted way, it is evaluation strategies that determine
the way beta-reduction is applied when the lambda-calculus is used as a programming language.
Evaluation strategies often imply how intermediate results are copied, discarded,
cached or reused.
For example, everything is repeatedly evaluated as many times as
requested in the call-by-name strategy.
In the call-by-need strategy, once a function requests its input, the
input is evaluated and the result is cached for later use.
The call-by-value strategy evaluates function input and caches the
result even if the function does not require the input.

The implementation of any evaluation strategy must be
correct, first of all, i.e.\ it has to produce results as stipulated by the
strategy.
Once correctness is assured, the next concern is efficiency.
One may prefer better space efficiency, or better time
efficiency, and it is well known that one can be traded off for the other.
For example, time efficiency can be improved by caching more
intermediate results, which increases space cost.
Conversely, bounding space requires repeating computations, which adds
to the time cost.
Whereas  correctness is well defined for any evaluation strategy, there
is a certain freedom in managing efficiency.
The challenge here is how to produce a unified framework which is flexible 
enough to analyse and guide the choices required by this trade-off.
Recent studies by Accattoli et al.\
\cite{AccattoliDL16,AccattoliBM14,Accattoli16} clearly establish classes of efficiency
for abstract machines that implement a given evaluation strategy.
They characterise efficiency by means of the
number of beta-reduction applications required by the strategy, and
introduce two efficiency classes, namely ``efficient''
and ``reasonable''.
This classification of abstract machines
gives us a starting point to quantitatively
analyse the trade-offs required in an implementation.


\subsection{Token-Passing GoI}

We employ Girard's Geometry of Interaction (GoI)
\cite{Girard89GoI1}, a semantics of linear logic proofs, 
as a framework for  studying the trade-off
between time and space efficiency.
In particular we focus on the token-passing style of GoI, which gives
abstract machines for the lambda-calculus, pioneered by Danos and
Regnier~\cite{DanosR96} and Mackie~\cite{Mackie95}.
These machines evaluate a term of the lambda-calculus by 
translating the term to a graph, a network of simple transducers, which 
executes by passing a data-carrying token around.

Token-passing  GoI  decomposes higher-order
computation into local token actions, or low-level interactions of
simple components.
It can give strikingly innovative implementation techniques for
functional programs, such as Mackie's \emph{Geometry of
Implementation} compiler~\cite{Mackie95}, Ghica's \emph{Geometry of
Synthesis} (GoS) high-level synthesis tool~\cite{Ghica07GoS1}, and
Sch{\"{o}}pp's resource-aware program transformation to a low-level
language~\cite{Schoepp14a}.
The interaction-based approach is also convenient for the complexity
analysis of programs, e.g.\
Dal Lago and Sch{\"{o}}pp's \textsc{IntML} type system of
logarithmic-space evaluation~\cite{DalLagoS16}, and Dal Lago et al.'s
linear dependent type system of polynomial-time
evaluation~\cite{DalLagoG11,DalLagoP12}.

Fixed-space execution is essential for GoS, since in the case of digital circuits the memory footprint of the program must be known at compile-time, and fixed. Using a restricted version of the call-by-name language Idealised Algol~\cite{GhicaS11GoS3} not only the graph, but also the token itself can be given a fixed size. Surprisingly, this technique also allows the compilation of recursive programs~\cite{GhicaSS11GoS4}. The GoS compiler shows both the usefulness of the GoI as a guideline for unconventional compilation and the natural affinity between its space-efficient abstract machine and call-by-name evaluation. The practical considerations match the prior theoretical understanding of this connection~\cite{DanosR96}.

The token passed around a graph simulates graph rewriting
without actually rewriting, which is in fact an extremal instance of
the trade-off we mentioned above.
Token-passing GoI keeps the underlying graph fixed
and uses the data stored in the token to route it.
It therefore favours space efficiency at the cost of time efficiency.
The same computation is repeated when, instead, intermediate
results could have been cached by saving copies of certain sub-graphs
representing values.

\subsection{Interleaving Token Passing with Graph Rewriting}

Our intention is to lift the token-passing GoI to a framework to
analyse the trade-off of efficiency, by strategically interleaving it
with graph rewriting. 
We present the framework as an abstract machine that interleaves token
passing with graph rewriting.
The machine, called the \textit{Dynamic GoI Machine} (DGoIM), is
defined as a state transition system with transitions for token
passing as well as transitions for graph rewriting. 
The key idea is that the token holds control over graph rewriting, by
visiting redexes and triggering the rewrite transitions.

Graph rewriting offers fine control over caching and sharing
intermediate results.
Through graph rewriting, the DGoIM can reduce sub-graphs visited by the
token, avoiding repeated token actions and improving time efficiency.
However, fetching cached results can increase the size of the graph.
In short, introduction of graph rewriting sacrifices space while
favouring time efficiency.
We expect the flexibility given by a fine-grained control over
interleaving will enable a careful balance between space and time
efficiency.

As a first step in our exploration of the flexibility of this machine,
we consider the two extremal cases of interleaving.
The first extremal case is ``passes-only'', in which the DGoIM never
triggers graph rewriting, yielding an ordinary token-passing abstract
machine.
As a typical example,
the $\lambda$-term $\app{(\abs{x}{}{t})}{u}$ is evaluated like
this:

\vspace{1.5ex}
\begin{center}
 \begin{minipage}[c]{.15\linewidth}
  \centering
  \begin{tikzpicture}[baseline=(current bounding box.west)]
   \draw [rounded corners, very thick] (-6.5,0) rectangle (-0.5,3)
   node [midway] {$\abs{x}{}{t}$};
   \draw [rounded corners, very thick] (0.5,0) rectangle (6.5,3)
   node [midway] {$u$};
   \fill (0,-2) circle [radius = 0.4];
   \draw [rounded corners, very thick]
   (0,-2)--(-3.5,-1.5)--(-3.5,0)
   (0,-2)--(3.5,-1.5)--(3.5,0)
   (0,-2)--(0,-4);
  \end{tikzpicture}
 \end{minipage} \hspace{.025\linewidth} %
 \begin{minipage}[c]{.75\linewidth}
  \begin{enumerate}
   \item A token enters the graph on the left at the bottom open edge.
   \item A token visits and goes through the left sub-graph
	 $\abs{x}{}{t}$.
   \item Whenever a token detects an occurrence of the variable
	 $x$ in $t$, it traverses the right sub-graph $u$, then
	 returns carrying information about the resulting value of
	 $u$.
   \item A token finally exits the graph at the bottom open
	 edge.
  \end{enumerate}
 \end{minipage}
\end{center}
\vspace{1.5ex}

\noindent
Step 3 is repeated whenever the argument $u$ needs to be
re-evaluated. This passes-only strategy of interleaving corresponds to
call-by-name evaluation.

The other extreme is ``rewrites-first'', in which
the DGoIM interleaves token passing with as much, and as early, graph
rewriting as possible, guided by the token. This corresponds to both
call-by-value and call-by-need evaluations, with different
trajectories of the token.
In the case of left-to-right call-by-value, the token enters the graph
from the bottom, traverses the left-hand-side sub-graph, which happens
to be already a value, then visits the sub-graph $u$ even before the
bound variable $x$ is used in a call.
The token causes rewrites while traversing the sub-graph $u$, and when
it exits, it leaves behind a graph corresponding to a value $v$ such
that $u$ reduces to $v$.
For right-to-left call-by-value, the token visits the sub-graph $u$
straightaway after entering the whole graph, reduces the sub-graph
$u$, to the graph of the value $v$, and visits the left-hand-side
sub-graph.
The difference with call-by-need is that the token visits and reduces
the sub-graph $u$ only when the variable $x$ is encountered in
$\abs{x}{}{t}$.

In our framework, all these three evaluations involve similar tactics
for caching intermediate results.
Different trajectories of the token realise their only difference,
which is the timing of cache creation.
Cached values are fetched in the same way: namely, if repeated
evaluation is required, then the sub-graph corresponding now to the
value $v$ is copied.
One copy can be further rewritten, if needed, while the original is
kept for later reference.

\subsection{Contributions}

This work presents a token-guided graph-rewriting abstract machine for
call-by-need, left-to-right call-by-value, and right-to-left
call-by-value evaluations.
The abstract machine is given by the rewrites-first strategy of the
DGoIM, which turns out to be as natural as the passes-only strategy for
call-by-name evaluation.
It switches the evaluations, by simply having different nodes that
correspond to the three different evaluations, rather than modifying
the behaviour of a single node to suite different evaluation demands.
This can be seen as a case study illustrating the flexibility of the
DGoIM,
which is achieved through controlled interleaving of rewriting and
token-passing, and through changing graph representations of terms.

We prove the soundness and completeness of the extended machine with
respect to the three evaluations separately, using a ``sub-machine''
semantics, where the word ``sub'' indicates both a focus on
substitution and its status as an intermediate representation.
The sub-machine semantics is based on Sinot's ``token-passing''
semantics~\cite{Sinot05,Sinot06} that makes explicit the two main
tasks of abstract machines: searching redexes and substituting
variables.

The time-cost analysis classifies the machine as ``efficient'' in
Accattoli's taxonomy of abstract machines~\cite{Accattoli16}.
We follow Accattoli et al.'s general methodology for quantitative
analysis of abstract machines~\cite{AccattoliBM14,Accattoli16},
however the method cannot be used ``off the shelf''.
Our machine is a more refined transition system with more transition
steps, and therefore does not satisfy one of their
assumptions~\cite[Sec.~3]{Accattoli16}, which requires one-to-one
correspondence of transition steps.
We overcome this technical difficulty by building a weak simulation of
the sub-machine semantics, which is also used in the proof of
soundness and completeness.
The sub-machine semantics resembles Danvy and Zerny's storeless
abstract machine~\cite{DanvyZ13}, to which the general recipe of cost
analysis does apply.

Finally, an on-line visualiser\footnote{
Link to the on-line visualiser:
\url{https://koko-m.github.io/GoI-Visualiser/}}
is implemented, in which our machine can be executed on arbitrary
closed (untyped) lambda-terms.
The visualiser also supports an existing abstract
machine based on the token-passing GoI, which will be discussed
later, to illustrate various resource usage of abstract machines.

The rest of the paper is organised as follows.
We present the sub-machine semantics in Sec.~\ref{sec:TermCalculus},
and introduce the DGoIM with the rewrites-first strategy in
Sec.~\ref{sec:Trans}.
In Sec.~\ref{sec:impl-eval-strat}, we show how the DGoIM implements
the three evaluation strategies via translation of terms into graphs,
and establish a weak simulation of the sub-machine semantics by the
DGoIM. The simulation result is used to prove soundness and
completeness of the DGoIM, and to analyse its time cost, in
Sec.~\ref{sec:time-cost-analysis}.
We compare our graph-rewriting approach
to improve time efficiency of token-passing GoI, with another approach
from the literature, namely ``jumping'' approach, in
Sec.~\ref{sec:rewr-vs.-jump}.
We discuss related conventional abstract machines
in Sec.~\ref{sec:relat-work-concl}.




\section{A Term Calculus with Sub-Machine Semantics}
\label{sec:TermCalculus}

We use an untyped term calculus that accommodates three
evaluation strategies of the lambda-calculus, by dedicated constructors for function application: namely,
$\appLazy{}{}$ (call-by-need), $\appLR{}{}$ (left-to-right
call-by-value) and $\appRL{}{}$ (right-to-left call-by-value).
The term calculus uses all strategies so that we do not have to present three 
almost identical calculi. Nevertheless, we are not interested in their
interaction but in each strategy separately.
In the rest of the paper, we therefore assume that each term contains
function applications of a single strategy.
As shown in the top of Fig.~\ref{fig:MockMachine}, the calculus
accommodates explicit substitutions $[x \leftarrow u]$.
A term with no explicit substitutions is said to be ``pure''.

The sub-machine semantics is used to establish the soundness of the
graph-rewriting abstract machine.
It imitates an abstract machine, by having the following two features.
Firstly, it extends conventional reduction semantics with reduction
steps that explicitly search for a redex, following the style of
Sinot's ``token-passing semantics''~\cite{Sinot05,Sinot06}.
Secondly, it decomposes the meta-level substitution into on-demand
linear substitution, using explicit substitutions, as the linear
substitution calculus does~\cite{AccattoliK10}.
The sub-machine semantics
also resembles a storeless abstract machine (e.g.~\cite[Fig.~8]{DanvyMMZ12}).
However the semantics is still too ``abstract'' to be considered
an abstract machine, in the sense that it works modulo
alpha-equivalence to avoid variable captures.

Fig.~\ref{fig:MockMachine} defines the sub-machine semantics of our
calculus.
It is given by labelled relations between \emph{enriched} terms
$\plug{E}{\W{t}}$.
In an enriched term $\plug{E}{\W{t}}$, a sub-term $t$ is not plugged directly into the evaluation 
context, but into a ``window'' $\W{\cdot}$ which makes it syntactically 
obvious where the reduction context is situated. 
Forgetting the window turns an enriched term into an ordinary term.
Basic rules $\mapsto$ are labelled with $\beta$, $\sigma$ or
$\epsilon$.
The basic rules~(\ref{eq:2}),~(\ref{eq:5}) and~(\ref{eq:8}), labelled
with $\beta$, apply
beta-reduction and delay substitution of a bound variable.
Substitution is done one by one, and on demand, by the basic
rule~(\ref{eq:10}) with label~$\sigma$.
Each application of the basic rule~(\ref{eq:10}) replaces exactly one
bound
variable with a value, and keeps a copy of the value for later use.
All  other basic rules, with label $\epsilon$, search for a
redex by moving the window without changing the underlying term.
Finally, reduction is defined by congruence of basic rules with
respect to evaluation contexts, and labelled accordingly.
Any basic rules and reductions are indeed between enriched terms,
because the window $\W{\cdot}$ is never duplicated or discarded.
They are also deterministic.
\begin{figure}[t]
 \begin{align*}
  \text{Terms} &&
  t &\,::=\, x \mid \abs{x}{}{t}
  \mid \appLazy{t}{t} \mid \appLR{t}{t} \mid \appRL{t}{t}
  \mid t[x \leftarrow t] \\
  \text{Values} &&
  v &\,::=\, \abs{x}{}{t} \\
  \text{Answer contexts} &&
  A &\,::=\, \emptyCtxt \mid A[x \leftarrow t] \\
  \text{Evaluation contexts} &&
  E &\,::=\, \emptyCtxt
  \mid E[x \leftarrow t] \mid \plug{E}{x}[x \leftarrow E] \\
  && &\quad\enspace
  \mid \appLazy{E}{t}
  \mid \appLR{E}{t} \mid \appLR{\plug{A}{v}}{E}
  \mid \appRL{t}{E} \mid \appRL{E}{\plug{A}{v}}
 \end{align*}
 \begin{align}
  \text{Basic rules }
  \mapsto_\beta,\mapsto_\sigma,\mapsto_\epsilon: &&
  \W{\appLazy{t}{u}} &\mapsto_\epsilon \appLazy{\W{t}}{u}
  \label{eq:1}\\ &&
  \appLazy{\plug{A}{\W{\abs{x}{}{t}}}}{u}
  &\mapsto_\beta \plug{A}{\W{t}[x \leftarrow u]}
  \label{eq:2}\\ &&
  \W{\appLR{t}{u}} &\mapsto_\epsilon \appLR{\W{t}}{u}
  \label{eq:3}\\ &&
  \appLR{\plug{A}{\W{\abs{x}{}{t}}}}{u}
  &\mapsto_\epsilon \appLR{\plug{A}{\abs{x}{}{t}}}{\W{u}}
  \label{eq:4}\\ &&
  \appLR{\plug{A}{\abs{x}{}{t}}}{\plug{A'}{\W{v}}}
  &\mapsto_\beta \plug{A}{\W{t}[x \leftarrow \plug{A'}{v}]}
  \label{eq:5}\\ &&
  \W{\appRL{t}{u}} &\mapsto_\epsilon \appRL{t}{\W{u}}
  \label{eq:6}\\ &&
  \appRL{t}{\plug{A}{\W{v}}}
  &\mapsto_\epsilon \appRL{\W{t}}{\plug{A}{v}}
  \label{eq:7}\\ &&
  \appRL{\plug{A}{\W{\abs{x}{}{t}}}}{\plug{A'}{v}}
  &\mapsto_\beta \plug{A}{\W{t}[x \leftarrow \plug{A'}{v}]}
  \label{eq:8}\\ &&
  \plug{E}{\W{x}}[x \leftarrow \plug{A}{u}]
  &\mapsto_\epsilon \plug{E}{x}[x \leftarrow \plug{A}{\W{u}}]
  \notag \\ &&
  & \text{($u$ is not in the form of $\plug{A'}{t'}$)}
  \label{eq:9}\\ &&
  \plug{E}{x}[x \leftarrow \plug{A}{\W{v}}] &\mapsto_\sigma
  \plug{A}{\plug{E}{\W{v}}[x \leftarrow v]}
  \label{eq:10}\\
  \text{Reductions }
  \multimap_\beta,\multimap_\sigma,\multimap_\epsilon: &&
  \parbox[c]{2cm}{
  $\infer[(\chi \in \{ \beta,\sigma,\epsilon \})]
  {\plug{E}{\tilde{t}} \multimap_\chi \plug{E}{\tilde{u}}}
  {\tilde{t} \mapsto_\chi \tilde{u}}$
  } \notag
 \end{align}
 \caption{''Sub-Machine'' Operational Semantics}
 \label{fig:MockMachine}
\end{figure}

\begin{figure}[ht]
 \newcommand{\rw}{black}
 \newcommand{\ps}{gray}

 Call-by-need evaluation:

 \begin{tabular}{|c|}
  \hline \\ $
  \begin{aligned}
   \textcolor{\rw}{
   \W{\appLazy{(\id{x})}{(\appLazy{(\id{y})}{(\id{z})})}}}
   &\textcolor{\ps}{\multimap_\epsilon
   \appLazy{\W{\id{x}}}{(\appLazy{(\id{y})}{(\id{z})})}} \\
   &\textcolor{\rw}{\multimap_\beta
   \W{x}[x \leftarrow \appLazy{(\id{y})}{(\id{z})}]} \\
   &\textcolor{\ps}{\multimap_\epsilon
   x[x \leftarrow \W{\appLazy{(\id{y})}{(\id{z})}}]} \\
   &\textcolor{\ps}{\multimap_\epsilon
   x[x \leftarrow \appLazy{\W{\id{y}}}{(\id{z})}]} \\
   &\textcolor{\rw}{\multimap_\beta
   x[x \leftarrow \W{y}[y \leftarrow \id{z}]]} \\
   &\textcolor{\ps}{\multimap_\epsilon
   x[x \leftarrow y[y \leftarrow \W{\id{z}}]]} \\
   &\textcolor{\rw}{\multimap_\sigma
   x[x \leftarrow \W{\id{z}}[y \leftarrow \id{z}]]} \\
   &\textcolor{\rw}{\multimap_\sigma
   \W{\id{z}}[x \leftarrow \id{z}][y \leftarrow \id{z}]}
  \end{aligned}
  $ \\ \\ \hline
 \end{tabular}

 Call-by-value evaluations:

 \begin{tabular}{|cc|cc|}
  \hline &&& \\ $
  \begin{aligned}
   &\textcolor{\rw}{
   \W{\appLR{(\id{x})}{(\appLR{(\id{y})}{(\id{z})})}}} \\
   &\textcolor{\ps}{\multimap_\epsilon
   \appLR{\W{\id{x}}}{(\appLR{(\id{y})}{(\id{z})})}} \\
   &\textcolor{\ps}{\multimap_\epsilon
   \appLR{\id{x}}{\W{\appLR{(\id{y})}{(\id{z})}}}} \\
   &\textcolor{\ps}{\multimap_\epsilon
   \appLR{\id{x}}{(\appLR{\W{\id{y}}}{(\id{z})})}} \\
   &\textcolor{\ps}{\multimap_\epsilon
   \appLR{\id{x}}{(\appLR{(\id{y})}{\W{\id{z}}})}} \\
   &\textcolor{\rw}{\multimap_\beta
   \appLR{\id{x}}{(\W{y}[y \leftarrow \id{z}])}} \\
   &\textcolor{\ps}{\multimap_\epsilon
   \appLR{\id{x}}{(y[y \leftarrow \W{\id{z}}])}} \\
   &\textcolor{\rw}{\multimap_\sigma
   \appLR{\id{x}}{(\W{\id{z}}[y \leftarrow \id{z}])}} \\
   &\textcolor{\rw}{\multimap_\beta
   \W{x}[x \leftarrow (\id{z})[y \leftarrow \id{z}]]} \\
   &\textcolor{\ps}{\multimap_\epsilon
   x[x \leftarrow \W{\id{z}}[y \leftarrow \id{z}]]} \\
   &\textcolor{\rw}{\multimap_\sigma
   \W{\id{z}}[x \leftarrow \id{z}][y \leftarrow \id{z}]} \\
  \end{aligned}
  $ &&& $
  \begin{aligned}
   &\textcolor{\rw}{
   \W{\appRL{(\id{x})}{(\appRL{(\id{y})}{(\id{z})})}}} \\
   &\textcolor{\ps}{\multimap_\epsilon
   \appRL{\id{x}}{\W{\appRL{(\id{y})}{(\id{z})}}}} \\
   &\textcolor{\ps}{\multimap_\epsilon
   \appRL{\id{x}}{(\appRL{(\id{y})}{\W{\id{z}}})}} \\
   &\textcolor{\ps}{\multimap_\epsilon
   \appRL{\id{x}}{(\appRL{\W{\id{y}}}{(\id{z})})}} \\
   &\textcolor{\rw}{\multimap_\beta
   \appRL{\id{x}}{(\W{y}[y \leftarrow \id{z}])}} \\
   &\textcolor{\ps}{\multimap_\epsilon
   \appRL{\id{x}}{(y[y \leftarrow \W{\id{z}}])}} \\
   &\textcolor{\rw}{\multimap_\sigma
   \appRL{\id{x}}{(\W{\id{z}}[y \leftarrow \id{z}])}} \\
   &\textcolor{\ps}{\multimap_\epsilon
   \appRL{\W{\id{x}}}{((\id{z})[y \leftarrow \id{z}])}} \\
   &\textcolor{\rw}{\multimap_\beta
   \W{x}[x \leftarrow (\id{z})[y \leftarrow \id{z}]]} \\
   &\textcolor{\ps}{\multimap_\epsilon
   x[x \leftarrow \W{\id{z}}[y \leftarrow \id{z}]]} \\
   &\textcolor{\rw}{\multimap_\sigma
   \W{\id{z}}[x \leftarrow \id{z}][y \leftarrow \id{z}]} \\
  \end{aligned}
  $ \\ &&& \\ \hline
 \end{tabular}
 \caption{Evaluations of $\app{(\id{x})}{(\app{(\id{y})}{(\id{z})})}$}
 \label{fig:ExampleEval}
\end{figure}

An \emph{evaluation} of a pure term $t$ (i.e.\ a term with no explicit
substitution) is a sequence of reductions
starting from $\plug{}{\W{t}}$, which is simply $\W{t}$.
Fig.~\ref{fig:ExampleEval} shows evaluations of a pure
term $\app{(\id{x})}{(\app{(\id{y})}{(\id{z})})}$
in the three evaluation strategies.
Reductions labelled with $\beta$ and $\sigma$, which change an
underlying term, are highlighted in black.
All three evaluations involve two beta-reductions, which apply
$\id{x}$ and $\id{y}$ to an argument.
Application of $\id{x}$ comes first in
the call-by-need evaluation, and delayed application of $\id{y}$
happens inside an explicit substitution.
On the other hand, in two call-by-value evaluations,
application of $\id{y}$ comes first, and no reduction happens inside
an explicit substitution.
The two call-by-value evaluations differ only in the way the window is
moved around function application.

The following lemma enables us to follow the use of sub-terms of the
initial term $t$ during the evaluation.
\begin{lem}
 \label{lem:SubMachineInvariants}
 For any evaluation $\W{t} \multimap^* \plug{E'}{\W{t'}}$ starting
 from a pure closed term $t$, the term $t'$ is a sub-term of $t$.
 Moreover, the evaluation context $E'$ is given by the following
 restricted grammar:
 \begin{align*}
  \overline{A} &\,::=\, \emptyCtxt
  \mid \overline{A}[x \leftarrow \plug{\overline{A}}{u}], \\
  \overline{E} &\,::=\, \emptyCtxt
  \mid \overline{E}[x \leftarrow \plug{\overline{A}}{u}]
  \mid \plug{\overline{E}}{x}[x \leftarrow \overline{E}] \\
  &\quad\enspace
  \mid \appLazy{\overline{E}}{u}
  \mid \appLR{\overline{E}}{u}
  \mid \appLR{\plug{\overline{A}}{v}}{\overline{E}}
  \mid \appRL{u}{\overline{E}}
  \mid \appRL{\overline{E}}{\plug{\overline{A}}{v}}
 \end{align*}
 where $u$ and $v$ are sub-terms of $t$, and $v$ is
 additionally a value.
\end{lem}
\begin{proof}[Proof outline]
 The proof is by induction on the length $k$ of the
 evaluation $\W{t} \multimap^k \plug{E'}{\W{t'}}$.
 In the base case, where $k = 0$, we have
 $\overline{E} = \emptyCtxt$ and $t' = t$.
 The inductive case, where $k > 0$, is proved by inspecting a basic
 rule used in the last reduction of the evaluation.
 In the case of the basic rule~(\ref{eq:9}), the last reduction is in
 the form of
 $\plug{E_0}{\plug{E}{\W{x}}[x \leftarrow \plug{A}{u}]}
 \multimap_\epsilon
 \plug{E_0}{\plug{E}{x}[x \leftarrow \plug{A}{\W{u}}]}$
 where $u$ is not in the form of $\plug{A''}{t''}$.
 By induction hypothesis,
 $\plug{E_0}{\plug{E}{}[x \leftarrow \plug{A}{u}]}$ follows the
 restricted grammar, and in particular, $\plug{A}{u}$ can be
 decomposed into a restricted answer context and a
 sub-term of $t$. Because a sub-term of $t$ is also pure, it follows
 that $A$ itself is a restricted answer context and $u$ is a sub-term
 of $t$.
\end{proof}

\section{The Token-Guided Graph-Rewriting Machine}\label{sec:Trans}


In the initial presentation of this work~\cite{MuroyaG17}, we used
proof nets of the multiplicative and exponential fragment of linear
logic~\cite{Girard87LL} to implement the call-by-need evaluation
strategy.
Aiming additionally at two call-by-value evaluation
strategies, we here use graphs that are closer to syntax trees but
are still augmented with the $\oc$-box structure taken from proof
nets.
Moving towards syntax trees allows us to accommodate two call-by-value
evaluations in a uniform way.
The $\oc$-box structures specify duplicable sub-graphs,
and help time-cost analysis of implementations.


\subsection{Graphs with Interface}

We use directed graphs, whose nodes are classified
into \emph{proper} nodes and \emph{link} nodes.
Link nodes are required to meet the following conditions.
\begin{itemize}
 \item For each edge, at least one of its two endpoints is a link
       node.
 \item Each link node is a source of at most one edge, and a target of
       at most one edge.
\end{itemize}
In particular, a link node is called \emph{input} if it is not a
target of any edge, and \emph{output} if it is not a source of any
edge.
An \emph{interface} of a graph is given by the set of all inputs and
the set of all outputs.
When a graph $G$ has $n$ input link nodes and $m$ output link nodes,
we sometimes write $G(n,m)$ to emphasise its interface.
If a graph has exactly one input, we refer to the input link node as
``root''.

\begin{figure}[t]
 \centering
 \includegraphics[scale=0.7]{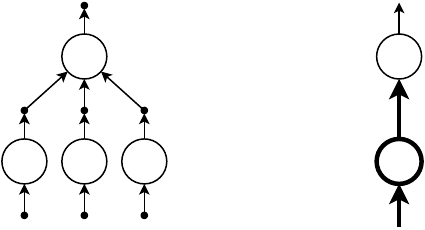}
 \caption{Full (Left) and Simplified (Right) Representation of a
 Graph $G(3,1)$}
 \label{fig:GraphExample}
\end{figure}
An example graph $G(3,1)$ is shown on the left in
Fig.~\ref{fig:GraphExample}. It has four proper nodes depicted by
circles, and seven link nodes depicted by bullets. Its three inputs
are placed at the bottom and one output is at the top.
Shown on the right in Fig.~\ref{fig:GraphExample} is a simplified
version of the representation. We use the following simplification
scheme:
not drawing link nodes explicitly (unless necessary), and
using a bold-stroke arrow (resp.\ circle) to represent a bunch of
parallel edges (resp.\ proper nodes).

The idea of using link nodes, as distinguished from proper nodes, 
comes from a graphical formalisation of string
diagrams~\cite{KissingerPhD}.\footnote{
Our link nodes should not be confused with the terminology ``link'',
which refers to a counterpart of our proper nodes, of proof nets.
}
String diagrams consist of ``boxes'' that are connected to each other
by ``wires'', and may have dangling or looping wires.
In the formalisation, boxes are modelled by ``box-vertices''
(corresponding to proper nodes in our case), and wires are modelled by
consecutive edges connected via ``wire-vertices''
(corresponding to link nodes in our case).
It is link nodes that allow dangling or looping wires to be properly
modelled.
The segmentation of wires into edges can introduce an arbitrary number
of consecutive link nodes, however these consecutive link nodes are
identified by the notion of ``wire homeomorphism''.
We will later discuss these consecutive link nodes, from the
perspective of the graph-rewriting machine.
From now on we simply call a proper node ``node'', and a link node
``link''.

Finally, an operation $\circ_{n,m}$ on graphs, parametrised by natural
numbers $n$ and $m$, is defined as follows:
\begin{equation*}
 G \circ_{n,m} H :=
  \parbox[c]{5cm}{\includegraphics[scale=0.9]{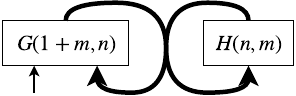}}.
\end{equation*}
In the sequel, we omit the parameters $n,m$ and simply write $\circ$.

\subsection{Node Labels and $\oc$-Boxes}

We use the following set $\mathcal{L}$ to label nodes:
\begin{equation*}
 \mathcal{L} = \{ \lambda, \appLazy{}{}, \appLR{}{}, \appRL{}{},
  \oc, \wn, \lb{D} \}
  \cup \{ \lb{C}_n \mid \text{$n$: a natural number} \}.
\end{equation*}
A node labelled with $X \in \mathcal{L}$ is called an ``$X$-node''.
The first four labels
correspond to the constructors of the calculus presented in
Sec.~\ref{sec:TermCalculus}, namely $\lambda$
(abstraction), $\appLazy{}{}$ (call-by-need application), $\appLR{}{}$
(left-to-right call-by-value application) and $\appRL{}{}$
(right-to-left call-by-value application).
These three application nodes are the novelty of this work.
The token, travelling in a graph, reacts to these nodes in
different ways, and hence implements different evaluation orders.
We believe that this is a more extensible way to accommodate different
evaluation orders, than to let the token react to the same node in
different ways depending on situation.

\begin{wrapfigure}[8]{r}{.5\linewidth}
  \centering
  \vspace{-2ex}
  \includegraphics[scale=0.8]{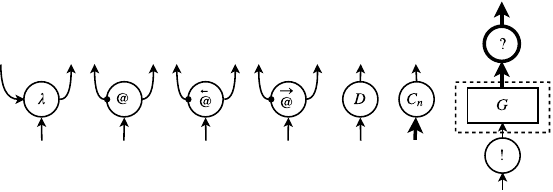}
  \vspace{-2ex}
  \caption{Generators of Graphs}\label{fig:connection}
\end{wrapfigure}
The other labels, namely $\oc$, $\wn$, $\lb{D}$ and $\lb{C}_n$ for
any natural number $n$, are used in the management of copying
sub-graphs. These are inspired by proof nets of the multiplicative and
exponential fragment of linear logic~\cite{Girard87LL}, and
$\lb{C}_n$-nodes generalise the standard binary contraction and
subsume weakening.

We use the \emph{generators} in Fig.~\ref{fig:connection} to build
labelled graphs.
Most generators are given by a graph that consists of one node and a
fixed number of adjacent links.
The number of input/output and incoming/outgoing edges for the node is
determined by the label, as indicated in the figure; in particular, a
label $\lb{C}_n$ indicates $n$ inputs and $n$ incoming edges.
We distinguish two outputs of an application node
($\appLazy{}{}$, $\appLR{}{}$ or $\appRL{}{}$), calling one
``function output'' and the other ``argument output''
(cf.~\cite{AccattoliG09}).
A bullet $\bullet$ in the figure specifies a function output.

The last generator in Fig.~\ref{fig:connection}
turns a graph $G(1,m)$ into a sub-graph (``$\oc$-box''),
by connecting it to one $\oc$-node (``principal
door'') and $m$ $\wn$-nodes (``auxiliary doors'').
This $\oc$-box structure is indicated by a dashed box in the figure.
The $\oc$-box structure, taken from proof nets, assists the
management of duplication of
sub-graphs by specifying those that can be copied.\footnote{
Our formalisation of graphs is based on the view of proof nets as
string diagrams, and hence of $\oc$-boxes as functorial
boxes~\cite{Mellies06}.
}

\subsection{Graph States and Transitions}

We define a graph-rewriting abstract machine as a labelled transition
system between \emph{graph states}.
\begin{defi}[Graph states]
 A \emph{graph state} $((G(1,0),e),\delta)$ is formed of a graph
 $G(1,0)$ with its distinguished link $e$, and token data
 $\delta = (d,f,S,B)$ that consists of:
 \begin{itemize}
  \item a \emph{direction} defined by
	$d ::= \up \mid \dn$,
  \item a \emph{rewrite flag} defined by
	$f ::= \square \mid \lambda \mid \oc$,
  \item a \emph{computation stack} defined by
	$S ::= \square \mid \star:S \mid \lambda:S \mid @:S$, and
  \item a \emph{box stack} defined by
	$B ::= \square \mid \star:B \mid \oc:B \mid
	\diamond:B \mid e':B$,
	where $e'$ is any link of the graph $G$.
 \end{itemize}
\end{defi}
\noindent
The distinguished link $e$  is called the
``position'' of the token.
The token reacts to a node in a graph using its data, which determines its
path.
Given a graph $G$ with root $e_0$,
the \emph{initial} state $\Init(G)$ on it is given by
$((G,e_0),(\up,\square,\square,\star:\square))$, and the \emph{final}
state $\Final(G)$ on it is given by
$((G,e_0),(\dn,\square,\square,\oc:\square))$.
An \emph{execution} on a graph $G$ is a sequence of transitions
starting from the initial state $\Init(G)$.

Each transition
$((G,e),\delta) \to_\chi ((G',e'),\delta')$ between graph states is
labelled by either $\beta$, $\sigma$ or $\epsilon$.
Transitions are deterministic, and classified into \emph{pass}
transitions that search for redexes and trigger rewriting, and 
\emph{rewrite} transitions that actually rewrite a graph as soon as a
redex is found.

A pass transition
$((G \circ H,e),(d,\square,S,B)) \to_\epsilon
 ((G \circ H,e'),(d',f',S',B'))$,
always labelled with $\epsilon$, applies to a state whose rewrite flag
is $\square$.
The graph $H$ contains only one node, and the positions $e$ and $e'$
are an input or an output of the node.
The transition
simply moves the token over the node, and updates its data by
modifying the top elements of stacks, while keeping an underlying
graph unchanged.
When the token passes a $\lambda$-node or a $\oc$-node, a rewrite
flag is changed to $\lambda$ or $\oc$, which triggers a rewrite
transition.
Fig.~\ref{fig:PassTransitions} defines pass transitions, by showing
the single-node graph $H$, token positions and data, omitting
the graph $G$.
The position of the token is drawn as a black triangle, pointing
towards the direction of the token.
The pass transition over a $\lb{C}_n$-node, where $n$ is positive,
pushes the old position $e$ to a box stack.
The link $e$ is drawn as a bullet.
\begin{figure}[t]
 \centering
 \includegraphics[scale=0.9]{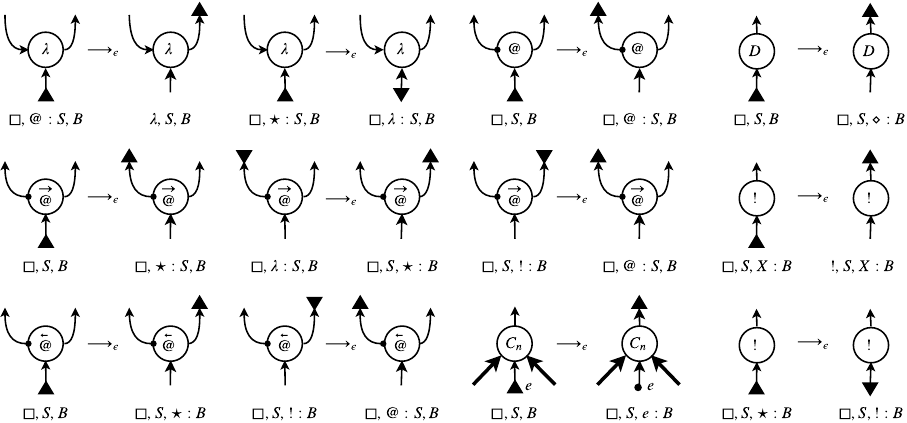}

 \vspace{2ex}
 where $X \neq \star$.
 \caption{Pass Transitions}
 \label{fig:PassTransitions}
\end{figure}

The way the token reacts to application nodes ($\appLazy{}{}$,
$\appLR{}{}$ and $\appRL{}{}$) corresponds to the way the window
$\W{\cdot}$ moves in evaluating these function applications in the
sub-machine semantics (Fig.~\ref{fig:MockMachine}).
When the token moves on to the composition output of an application
node, the top element of a computational stack is either $@$ or
$\star$.
The element $\star$ makes the token return  from a $\lambda$-node,
which corresponds to reducing the function part of application to a
value (i.e.\ abstraction).
The element $@$ lets the token proceed at a $\lambda$-node, raises the
rewrite flag $\lambda$, and hence triggers a rewrite transition that
corresponds to beta-reduction.
The call-by-value application nodes ($\appLR{}{}$ and $\appRL{}{}$)
send the token to their argument output,  pushing the element
$\star$ to a box stack.
This makes the token bounce at a $\oc$-node and return 
to the application node, which corresponds to evaluating the argument
part of function application to a value.
Finally, pass transitions through $\lb{D}$-nodes, $\lb{C}_n$-nodes and
$\oc$-nodes prepare copying of values, and eventually raise the
rewrite flag $\oc$ that triggers on-demand duplication.

A rewrite transition
$((G \circ H,e),(d,f,S,B)) \to_\chi
 ((G \circ H',e'),(d',f',S,B'))$,
labelled with $\chi \in \{ \beta,\sigma,\epsilon \}$, applies to a
state whose rewrite flag is either $\lambda$ or $\oc$.
It replaces the sub-graph $H$ (``redex'') with the graph $H'$ of the
same interface.
The position $e$ that belongs to $H$ is changed to the position $e'$
that belongs to $H'$.
The transition may pop an element from a box stack.
Fig.~\ref{fig:RewriteTransitions} defines rewrite transitions, by
showing the sub-graphs $H$ and $H'$, as well as token positions and
data, omitting the graph $G$.
Before we go through each rewrite transition, we note that rewrite
transitions are not exhaustive in general, as a graph may not match a
redex even though a rewrite flag is raised.
However we will see that there is no failure of transitions in
implementing the term calculus.
\begin{figure}[t]
 \centering
 \includegraphics[scale=0.7]{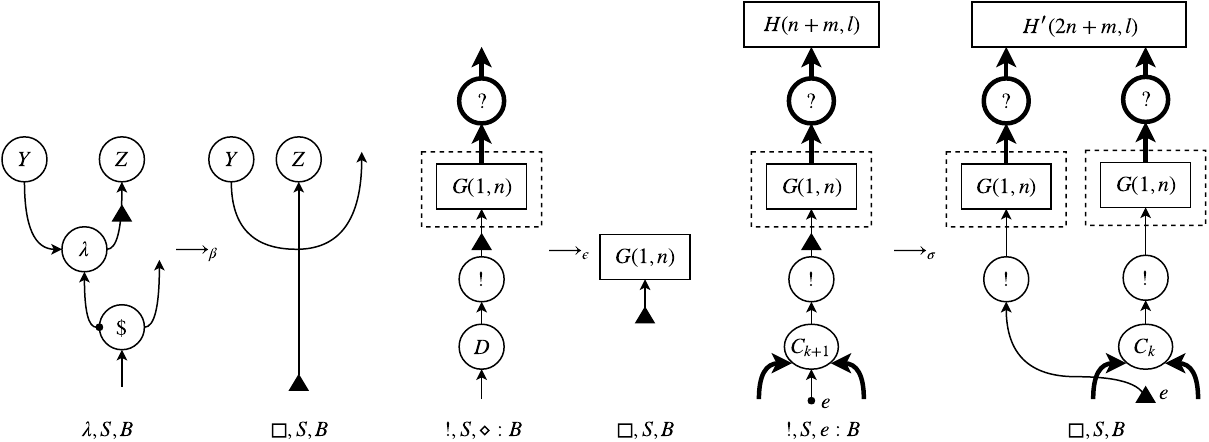}

 \vspace{2ex}
 where $Y \in \mathcal{L}$, $Z \in \mathcal{L}$,
 $\$ \in \{ \appLazy{}{},\appLR{}{},\appRL{}{} \}$, and
 $G(1,n)$ is any graph.
 \caption{Rewrite Transitions}
 \label{fig:RewriteTransitions}
\end{figure}

The first rewrite transition in Fig.~\ref{fig:RewriteTransitions},
with  label $\beta$, occurs when a rewrite flag is $\lambda$.
It implements beta-reduction by eliminating a pair of an abstraction
node ($\lambda$) and an application node
($\$ \in \{ \appLazy{}{},\appLR{}{},\appRL{}{} \}$ in the figure).
Outputs of the $\lambda$-node are required to be connected to
arbitrary nodes (labelled with $Y$ and $Z$ in the figure), so that
edges between links are not introduced.
The other rewrite transitions are for the rewrite flag $\oc$, and
they together realise the copying process of a sub-graph (namely a
$\oc$-box).
The second rewrite transition in Fig.~\ref{fig:RewriteTransitions},
labelled with $\epsilon$, finishes off each copying process by
eliminating all doors of the $\oc$-box $G$.
It replaces the interface of $G$ with output links of the auxiliary
doors and the input link of the $\lb{D}$-node, which is the new
position of the token, and pops the top element
$\diamond$ of a box stack.
Again, no edge between links are introduced.

The last rewrite transition in the figure, with label $\sigma$,
actually copies a $\oc$-box.
It requires the top element $e$ of the old box stack to be one of
input links of the $\lb{C}_{k+1}$-node (where $k$ is a natural
number).
The link $e$ is popped from the box stack and becomes the new position
of the token, and the $\lb{C}_{k+1}$-node becomes a
$\lb{C}_k$-node by keeping all the inputs except for the link $e$.
The sub-graph $H(n+m,l)$ must consist of $l$ parallel $\lb{C}$-nodes
that altogether have $n+m$ inputs.
Among these inputs, $n$ must be connected to auxiliary doors of the
$\oc$-box $G(1,n)$, and $m$ must be connected to nodes that are not in
the redex.
The sub-graph $H(n+m,l)$ is turned into $H'(2n+m,l)$ by introducing $n$
inputs to these $\lb{C}$-nodes as follows: if an auxiliary
door of the $\oc$-box $G$ is connected to a $\lb{C}$-node in $H$, two
copies of the auxiliary door are both connected to the corresponding
$\lb{C}$-node in $H'$.
Therefore the two sub-graphs consist of the same number $l$ of
$\lb{C}$-nodes, whose indegrees are possibly increased.
The $m$ inputs, connected to nodes outside a redex, are kept
unchanged.
Fig.~\ref{fig:CopyExample} shows an example where copying of the graph
$G(1,3)$ turns the graph $H(5,2)$ into $H'(8,2)$.
\begin{figure}[t]
 \centering
 \includegraphics[scale=0.9]{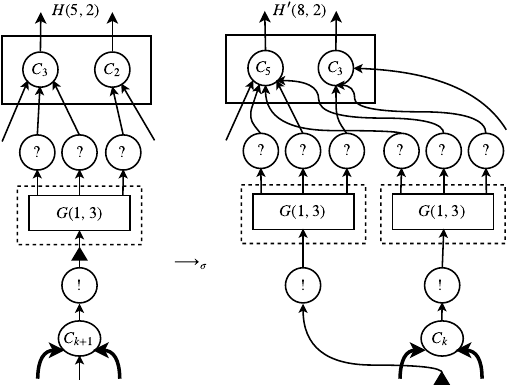}
 \caption{Example of Rewrite Transition $\to_\sigma$}
 \label{fig:CopyExample}
\end{figure}

All pass and rewrite transitions are well-defined, and indeed
deterministic. Pass transitions are also reversible, in the sense
that no two different pass transitions result in the same graph state.
No transition is possible at a final state, and no pass transition
results in an initial state.
An execution of pass transitions only has some continuity in the
following sense.
\begin{lem}[Pass continuity]
 \label{lem:Continuity}
 For any execution $\Init(G) \to^* ((G,e),\delta)$ of pass
 transitions only, there exists a non-empty sequence
 $e_1,\ldots,e_n$ of links of $G$ that satisfies the following.
 \begin{itemize}
  \item $e_1$ is the root of $G$, and $e_n = e$.
  \item For each $i \in \{ 1,\ldots,n-1 \}$, there exists a node whose
	inputs include $e_i$ and whose outputs include $e_{i+1}$.
  \item Each link in the sequence appears as a token position in the
	execution $\Init(G) \to^* ((G,e),\delta)$.
 \end{itemize}
\end{lem}
\begin{proof}[Proof outline]
 The proof is by induction on the length $k$ of the execution
 $\Init(G) \to^* ((G,e),\delta)$.
 In the base case, where $k = 0$, the link $e$ is the root of $G$, and
 $e$ itself as a sequence satisfies the conditions.
 The inductive case, where $k > 0$, is proved by inspecting all
 possibilities of the last pass transition in the sequence.
\end{proof}

The following
``sub-graph'' property is essential in time-cost analysis, because it
bounds the size of duplicable sub-graphs (i.e.\ $\oc$-boxes) in an
execution.
\begin{lem}[Sub-graph property]
 \label{lem:SubGraph}
 For any execution $\Init(G) \to^* ((H,e),\delta)$, each
 $\oc$-box of the graph $H$ appears as a sub-graph of the initial
 graph $G$.
\end{lem}
\begin{proof}
 Rewrite transitions can only copy or discard a $\oc$-box, and cannot
 introduce, expand or reduce a single $\oc$-box. Therefore, any
 $\oc$-box of $H$ has to be already a $\oc$-box of the initial graph
 $G$.
\end{proof}


When a graph has an edge between links, the token is just passed along.
With this pass transition over a link at hand, the equivalence relation between
graphs that identifies consecutive links with a single
link---so-called ``wire homeomorphism''~\cite{KissingerPhD}---lifts to
a weak bisimulation between graph states.
Therefore, behaviourally, we can safely ignore consecutive links.
From the perspective of time-cost analysis, we benefit from the fact
that rewrite transitions are designed not to introduce any edge between
links.
This means, by assuming that an execution starts with a graph with no
consecutive links, we can analyse time cost of the execution without
caring the extra pass transition over a link.

\section{Implementation of Evaluation Strategies}
\label{sec:impl-eval-strat}

The implementation of the term calculus, by means of the dynamic GoI,
starts with translating (enriched) terms into graphs.
The definition of the translation uses multisets of variables, to
track how many times each variable occurs in a term. We assume that terms
are alpha-converted in a form in which all binders introduce distinct
variables.
\begin{nota}[Multiset]
 We write $x \in^k M$ if the multiplicity of $x$ in a multiset $M$ is
 $k$.
 The empty multiset is denoted by $\emptyset$.
 The sum of two  multisets $M_1$ and $M_2$, denoted by $M_1 + M_2$, is
 defined as follows: $x \in^k {M_1 + M_2}$ if there exist $k_1$ and
 $k_2$ such that $x \in^{k_1} M_1$, $x \in^{k_2} M_2$ and
 $k = k_1 + k_2$.
 Removing \emph{all} $x$ from a multiset $M$ yields the multiset
 $M \backslash x$, e.g.\ ${[x,x,y] \backslash x} = [y]$.
 We abuse the notation and refer to a multiset $[x,\ldots,x]$ of a
 finite number of $x$'s, simply as $x$.
\end{nota}
\begin{defi}[Free variables]
 The map $\FV$ of terms to multisets of variables is inductively
 defined as below, where
 $\appGen{}{} \in \{ \appLazy{}{},\appLR{}{},\appRL{}{} \}$:
 \begin{align*}
  \FV(x) &:= [x], &
  \FV(\abs{x}{}{t}) &:= \FV(t) \backslash x, \\ 
  \FV(\appGen{t}{u}) &:= \FV(t) + \FV(u), &
  \FV(t[x \leftarrow u]) &:= (\FV(t) \backslash x) + \FV(u).
 \end{align*}
For a multiset $M$ of variables,
the map $\FV_M$ of
 evaluation contexts to multisets of variables is
 defined by:
 \begin{align*}
  \FV_M(\emptyCtxt) &:= M, \\
  \FV_M(\appLazy{E}{t}) &:= \FV_M(E) + \FV(t), \\
  \FV_M(\appLR{E}{t}) &:= \FV_M(E) + \FV(t), \\
  \FV_M(\appLR{\plug{A}{v}}{E}) &:= \FV(\plug{A}{v}) + \FV_M(E), \\
  \FV_M(\appRL{t}{E}) &:= \FV(t) + \FV_M(E), \\
  \FV_M(\appRL{E}{\plug{A}{v}}) &:= \FV_M(E) + \FV(\plug{A}{v}), \\
  \FV_M(E[x \leftarrow t]) &:=
  (\FV_M(E) \backslash x) + \FV(t), \\
  \FV_M(\plug{E'}{x}[x \leftarrow E]) &:=
  (\FV(\plug{E'}{x}) \backslash x) + \FV_M(E).
 \end{align*}
\end{defi}

A term $t$ is said be \emph{closed} if $\FV(t) = \emptyset$.
Consequences of the above definition are the following equations.
\begin{align*}
 \FV(\plug{E}{t}) &= \FV_{\FV(t)}(E), \\
 \FV_M(\plug{E}{E'}) &= \FV_{\FV_M(E')}(E), \\
 \FV_{M+M'}(E) &= \FV_M(E) + M'
 \quad \text{(if $M'$ is not captured in $E$)}, \\
 \FV_x(E) \backslash x &= \FV_\emptyset(E) \backslash x.
\end{align*}

We give translations of terms, answer contexts, and evaluation
contexts separately.
Fig.~\ref{fig:TermAnsCtxtTranslation} and
Fig.~\ref{fig:EvalCtxtTranslation} define
two mutually recursive translations $(\cdot)^\dag$ and $(\cdot)^\ddag$, the 
first one for terms and answer contexts, and the second one for
evaluation contexts.
In the figures,
${\appGen{}{}} \in \{ \appLazy{}{},\appLR{}{},\appRL{}{} \}$,
and $m$ is the multiplicity of $x$.
Fig.~\ref{fig:TranslationGeneral} shows the general form of the
translations, and
Fig.~\ref{fig:TranslationExample} shows translation of a term
$\appLazy{(\appLazy{
(\abs{f}{}{\abs{x}{}{\appLazy{f}{(\appLazy{f}{x})}}})
}{(\id{y})})}{(\id{z})}$.

\begin{figure}[t]
 \begin{minipage}[c]{.45\linewidth}
  \centering
  \includegraphics[scale=0.8]{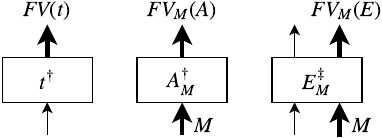}
  \caption{General Form of Translations}
  \label{fig:TranslationGeneral}
 \end{minipage} \hspace{.025\linewidth} %
 \begin{minipage}[c]{.5\linewidth}
  \centering
  \includegraphics[scale=0.6]{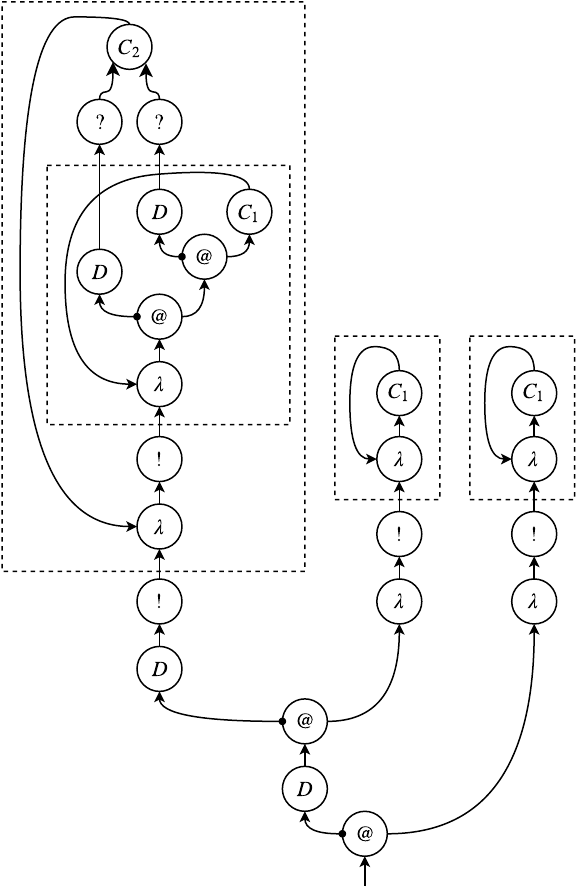}
  \caption[caption]{Translation of a Term \\
  {$\appLazy{(\appLazy{
  (\abs{f}{}{\abs{x}{}{\appLazy{f}{(\appLazy{f}{x})}}})
  }{(\id{y})})}{(\id{z})}$}}
  \label{fig:TranslationExample}
 \end{minipage}
\end{figure}

The DGoIM can evaluate a closed term $t$ by starting an execution on
the translation $t^\dag$. Executions on any translated closed pure
terms can be seen in our on-line
visualiser\footnote{\url{https://koko-m.github.io/GoI-Visualiser/}}.
The translations of answer contexts and evaluation contexts are to
define a weak simulation of the sub-machine semantics by the DGoIM,
which is then used to prove soundness, completeness and efficiency of
the DGoIM.

The annotation of bold-stroke edges means each edge of a bunch is
labelled with an element of the annotating multiset, in a one-to-one
manner.
In particular if a bold-stroke edge is annotated by a variable $x$,
all edges in the bunch are annotated by the variable $x$.
Translation $E^\ddag_M$ of an evaluation context has one input and one
output that are not annotated, which we refer to as the ``main'' input
and the ``main'' output. 
These annotations are only used to define the translations, and
are subsequently ignored during execution.

\begin{figure}[p]
 \begin{minipage}[c]{\linewidth}
  \centering
  \includegraphics[scale=0.9]{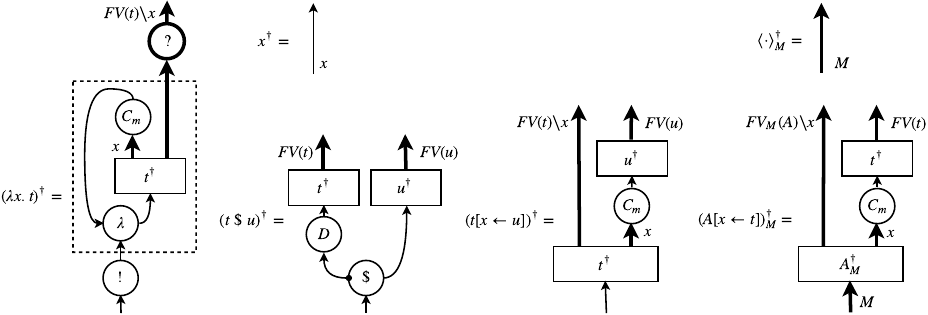}

    where
  $\$ \in \{ \appLazy{}{},\appLR{}{},\appRL{}{} \}$.
  \caption{Inductive Translation of Terms and Answer Contexts}
  \label{fig:TermAnsCtxtTranslation}
 \end{minipage}
 \begin{minipage}[c]{\linewidth}
  \centering
  \vspace{3ex}
  \includegraphics[scale=0.9]{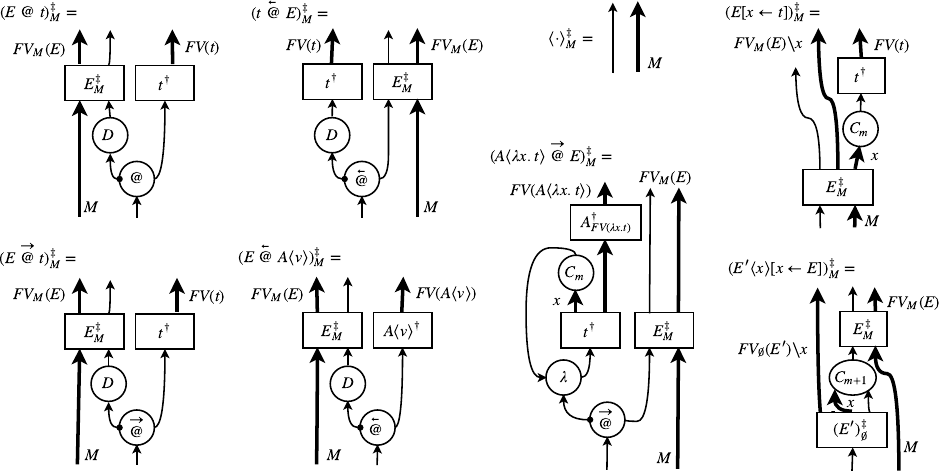}
  \caption{Inductive Translation of Evaluation Contexts}
  \label{fig:EvalCtxtTranslation}
 \end{minipage}
 \begin{minipage}[c]{\linewidth}
  \centering
  \vspace{3ex}
  \includegraphics[scale=0.7]{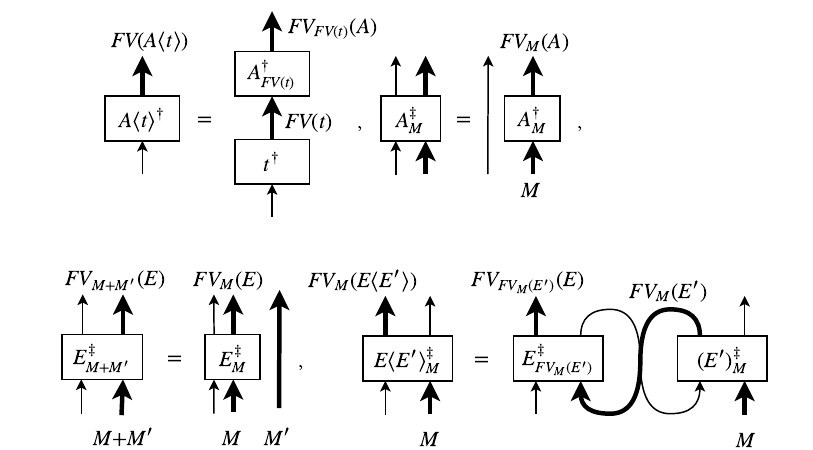}
  \caption{Decompositions of Translations}
  \label{fig:Decomposition}
 \end{minipage}
\end{figure}

The translations are based on the so-called ``call-by-value''
translation of linear logic to intuitionistic logic (e.g.\
\cite{MaraistOTW99}).
Only the translation of abstraction can be accompanied by a
$\oc$-box, which captures the fact
that only values (i.e.\ abstractions) can be duplicated
(see the basic rule~(\ref{eq:10}) in Fig.~\ref{fig:MockMachine}).
Note that only one $\lb{C}$-node is introduced for each bound
variable. This is vital to achieve constant cost in looking up a
variable, namely in realising the basic rule~(\ref{eq:9}) in
Fig.~\ref{fig:MockMachine}.

The two mutually recursive translations $(\cdot)^\dag$ and
$(\cdot)^\ddag$ are related by the decompositions in
Fig.~\ref{fig:Decomposition}, which can
be checked by straightforward induction.
In the third decomposition, $M'$ is not captured in $E$.
Note that, in general, the translation $\plug{E}{t}^\dag$ of a term in
an evaluation context cannot be decomposed into translations
$E^\ddag_{\FV(t)}$ and $t^\dag$.
This is because a translation
$(\appLR{\plug{A}{\abs{x}{}{t}}}{E})^\ddag_M$ lacks a $\oc$-box
structure, compared to a translation
$(\appLR{\plug{A}{\abs{x}{}{t}}}{u})^\dag$.

Translation of an evaluation context can be traversed by pass
transitions without raising the rewrite flag $\lambda$ or $\oc$, as
the following lemma states.
\begin{lem}
 \label{lem:PassEvalCtxt}
 Let $E$ be an evaluation context and $M$ be a multiset.
 For any graph $G(1,0)$ that has $E^\ddag_M$ as a sub-graph and has no
 edge between links,
 let $e_i$ and $e_o$ be the main input and the main output of the
 sub-graph $E^\ddag_M$, respectively.
 For any pair $(S,B)$ of a computation stack and a box stack,
 there exists a pair $(S',B')$ of a computation stack and a box stack,
 such that
 $((G,e_i),(\up,\square,S,B)) \to^* ((G,e_o),(\up,\square,S',B'))$
 is a sequence of pass transitions.
\end{lem}
\begin{proof}
 \newcommand{\pseq}{\overset{p}{\to}^*}

 By induction on $E$. We use $\pseq$ to denote a
 sequence of pass transitions in this proof.
 In the base case, where $E = \emptyCtxt$, the main input $e_i$ and
 the main output $e_o$ coincides. An empty sequence suffices.

 The first class of inductive cases are when the top-level constructor
 of $E$ is function application, e.g.\ $E \equiv \appLazy{E'}{t}$.
 Let $e'_i$ and $e'_o$ be the main input and the main output of the
 sub-graph $(E')^\ddag_M$, respectively.
 In each of the cases, there exist stacks $S''$ and $B''$ such that
 $((G,e_i),(\up,\square,S,B)) \pseq ((G,e'_i),(\up,\square,S'',B''))$.
 By the induction hypothesis, there exist stacks $S'$ and $B'$ such
 that
 $((G,e'_i),(\up,\square,S'',B'')) \pseq
 ((G,e'_o),(\up,\square,S',B'))$.
 Combining these two sequences yields a desired sequence, because
 $e'_o = e_o$.

 The inductive case where $E \equiv E'[x \leftarrow t]$ simply boils
 down to the induction hypothesis.

 The last inductive case is when
 $E \equiv \plug{E_1}{x}[x \leftarrow E_2]$.
 Let $e'_i$ and $e'_o$ be the main input and the main output of
 the sub-graph $(E_1)^\ddag_\emptyset$, and
 $e''_i$ and $e''_o$ be the main input and the main output of
 the sub-graph $(E_2)^\ddag_M$, respectively.
 We have $e_i = e'_i$ and $e_o = e''_o$. The link $e'_o$ is an
 input of a $\lb{C}$-node and $e''_i$ is the output of the
 $\lb{C}$-node.
 By the induction hypothesis on $E_1$, there exist stacks $S''$ and
 $B''$ such that
 $((G,e'_i),(\up,\square,S,B)) \pseq
 ((G,e'_o),(\up,\square,S'',B''))$.
 This sequence can be followed by a pass transition
 $((G,e'_o),(\up,\square,S'',B'')) \to
 ((G,e''_i),(\up,\square,S'',e'_o:B''))$.
 By the induction hypothesis on $E_2$, there exist stacks $S'$ and
 $B'$ such that
 $((G,e''_i),(\up,\square,S'',e'_o:B'')) \pseq
 ((G,e''_o),(\up,\square,S',B'))$.
 Combining all these sequences yields a desired sequence, because
 $e_i = e'_i$ and $e_o = e''_o$.
\end{proof}

The inductive translations lift to
a binary relation between closed enriched terms and graph states.
\begin{defi}[Binary relation $\preceq$]
 The binary relation $\preceq$ is defined by
 $\plug{E}{\W{t}} \preceq
 ((E^\ddag \circ t^\dag,e), (\up,\square,S,B))$, where:
 (i) $\plug{E}{\W{t}}$ is a closed enriched term, and
 $(E^\ddag \circ t^\dag,e)$ is given by
 \parbox[c]{3.2cm}{\includegraphics[scale=0.9]{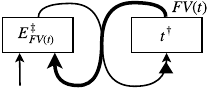}} with no
 edges between links,
 and (ii) there is an execution
 $\Init(E^\ddag \circ t^\dag) \to^*
 ((E^\ddag \circ t^\dag,e),(\up,\square,S,B))$
 of pass transitions only,
 in which $e$ appears as a token position only in the last state.
\end{defi}
A special case is $\W{t} \preceq \Init(t^\dag)$, which relates the
starting points of an evaluation and an execution.
We require the graph $E^\ddag \circ t^\dag$ to have no edges between
links, which is based on the discussion at the end of
Sec.~\ref{sec:Trans} and essential for time-cost analysis.
Although the definition of the translations uses edges between links
(e.g.\ the translation $x^\dag$), the graphs
$E^\ddag$ and $t^\dag$ can be constructed without introducing any edge
between links. For example, a variable can be translated into a single
link that is both an input and an output, and outputs of the
translation $(\appLazy{t}{u})^\dag$ can be simply the union of outputs
of $t^\dag$ and $u^\dag$.
The graph $E^\ddag \circ t^\dag$ can be constructed by
identifying interfaces of $E^\ddag$ and $t^\dag$, instead of
introducing edges.

%

The binary relation $\preceq$ gives a weak simulation of the
sub-machine semantics by the graph-rewriting machine.
The weakness, i.e.\ the extra transitions compared with reductions,
comes from the locality of pass transitions and the bureaucracy of
managing $\oc$-boxes.
\begin{thm}[Weak simulation with global bound]
 \label{thm:WeakSimulation}
 \noindent
 \begin{enumerate}
  \item If $\plug{E}{\W{t}} \multimap_\chi \plug{E'}{\W{t'}}$ and
	$\plug{E}{\W{t}} \preceq ((E^\ddag \circ t^\dag,e),\delta)$
	hold, then there exists a number $n \leq 3$ and	a graph state
	$(((E')^\ddag \circ (t')^\dag,e'),\delta')$ such that
	$((E^\ddag \circ t^\dag,e),\delta) \to_\epsilon^n \to_\chi
	(((E')^\ddag \circ (t')^\dag,e'),\delta')$
	and $\plug{E'}{\W{t'}} \preceq
	(((E')^\ddag \circ (t')^\dag,e'),\delta')$.
  \item If $\plug{A}{\W{v}} \preceq
	((A^\ddag \circ v^\dag,e),\delta)$ holds,
	then the graph state $((A^\ddag \circ v^\dag,e),\delta)$ is
	initial, from which only the transition
	$\Init(A^\ddag \circ v^\dag) \to_\epsilon
	\Final(A^\ddag \circ v^\dag)$
	is possible.
 \end{enumerate}
\end{thm}
\begin{proof}
 %
 For the second half, $e$ is the root of the graph
 $A^\ddag \circ v^\dag$, which means the state
 $((A^\ddag \circ v^\dag, e), \delta)$ is not a result of any pass
 transition. Therefore, by the condition (ii) of the binary relation
 $\preceq$, we have
 $\Init(A^\ddag \circ v^\dag) =
 ((A^\ddag \circ v^\dag, e), \delta)$,
 and one pass transition from this state yields a final state
 $\Final(A^\ddag \circ v^\dag)$.

 For the first half,
 Fig.~\ref{fig:Proof2}, Fig.~\ref{fig:Proof1} and
 Fig.~\ref{fig:Proof3}
 illustrate how the graph-rewriting machine
 simulates each reduction $\multimap$ of the sub-machine semantics.
 Each sequence of transitions $\to$ simulates a single reduction
 $\multimap$.
 Annotations of edges are omitted, and only the first and the last
 states of each sequence are shown, except for the case of the basic
 rule~(\ref{eq:10}).

 Some sequences involve equations that apply the four
 decomposition properties of the translations $(\cdot)^\dagger$ and
 $(\cdot)^\ddagger$, which are given earlier in this section.
 These equations rely on the fact that terms
 are alpha-converted in a form in which all binders introduce distinct
 variables, and reductions with
 labels $\beta$ and $\sigma$ work modulo alpha-equivalence to avoid
 name captures.
 This implies the following.
 \begin{itemize}
  \item Free variables of $u$ are not captured by $A$ in the case of
	the basic rule~(\ref{eq:2}).
  \item Free variables of $\plug{A'}{v}$ are not captured by $A$ in
	the case of the basic rules~(\ref{eq:5}) and~(\ref{eq:8}).
  \item The variable $x$ is not captured by $E$ or $E'$ in
	the case of the basic rules~(\ref{eq:9}) and~(\ref{eq:10}).
  \item In the case of the basic rule~(\ref{eq:10}),
	free variables of $E'$ are not captured by $A$,
	free variables of $v$ are not captured by $E'$,
	and $x$ does not freely appear in $v$.
 \end{itemize}

 Simulation of the basic rule~(\ref{eq:10}) involves duplicating the
 sub-graph $v^\dag$, which is a $\oc$-box.
 Because free variables of the value $v$ are captured by either $E$ or
 $A$, the multiset $\FV(v)$ can be partitioned into two multisets as
 $\FV(v) = M_E + M_A$, such that $M_E$ is the multiset of those
 captured by $E$ and $M_A$ is the multiset of those captured by $A$.
 No variable is contained by both $M_E$ and $M_A$.
 The translations $E^\ddag$ and $A^\dag$ include $\lb{C}$-nodes that
 correspond to $M_E$ and $M_A$, respectively. These $\lb{C}$-nodes
 get extra inputs by the rewrite transition labelled with $\sigma$,
 as represented by the middle state in the simulation sequence.

 In each sequence, let $G_s$ and $G_t$ be the first and the last
 graph, respectively.
 By the condition (ii) of the binary relation
 $\preceq$, there exists an execution
 $\ex : \Init(G_s) \to^* ((G_s, e_1), (\up, \square, S', B'))$ of only
 pass transitions, in which the link $e_1$ (see the figures) appears
 as a token position only once at the end.
 \begin{enumerate}
  \item \label{item:PassUpA}
	In simulation of the basic rules~(\ref{eq:1}),~(\ref{eq:3})
	and~(\ref{eq:6}),
	the figures use $S$ and $B$ instead of $S'$ and $B'$.
	By Lem.~\ref{lem:Continuity}, the result position $e_2$ (see
	the figures) does not appear in the execution $\ex$;
	if this is not the case, $e_1$ would appear more than once in
	$\ex$, which is a contradiction.
	Therefore, $\ex$ followed by the pass transitions shown in the
	figures gives a desired execution that meets the condition
	(ii) of the binary relation $\preceq$.
  \item \label{item:PassUpC}
	In simulation of the basic rule~(\ref{eq:9}),
	the figure uses $S$ and $B$ instead of $S'$ and $B'$.
	Because $x$ is not captured by $E'$, the starting position
	$e_1$ is in fact an input of the $\lb{C}_{m+1}$-node.
	Using Lem.~\ref{lem:Continuity} again in the same way, the
	result position $e_2$ does not appear in the execution $\ex$.
	Therefore, $\ex$ followed by the pass transition shown in the
	figures gives a desired execution that meets the condition
	(ii) of the binary relation $\preceq$.
  \item \label{item:PassDnUp}
	In simulation of the basic rule~(\ref{eq:7}),
	by the reversibility of pass transitions, there exist stacks
	$S$ and $B$ such that: $S' = S$, $B' = \star:B$, and
	the execution $\ex$ can be decomposed into an execution
	$\ex' : \Init(G_s) \to^* ((G_s, e_0), (\up, \square, S, B))$
	and one subsequent pass transition (see the figure for $e_0$).
	In the execution $\ex'$, the link $e_0$ appears as a token
	position only once at the end, which can be checked by
	contradiction as follows.
	\begin{itemize}
	 \item If $e_0$ appears more than once in $\ex'$ and its first
	       appearance is with direction $\dn$,
	       it must be a result of a pass transition.
	       However, no pass transition leads to this situation,
	       because $e_0$ is an input of a function application
	       node. This is a contradiction.
	 \item If $e_0$ appears more than once in $\ex'$ and its first
	       appearance is with direction $\up$,
	       it must be with rewrite flag $\square$, because $\ex'$
	       consists of pass transitions only.
	       Regardless of token data,
	       the first appearance leads to an extra appearance of
	       $e_1$ in $\ex'$, which is a contradiction.
	\end{itemize}
	Given this freshness of $e_0$ in $\ex'$, by
	Lem.~\ref{lem:Continuity}, the result position $e_2$ does not
	appear in the execution $\ex'$.
	Therefore, $\ex$ followed by the pass transitions shown in the
	figures gives a desired execution that meets the condition
	(ii) of the binary relation $\preceq$.
  \item \label{item:BetaRwEnd}
	In simulation of the basic rules~(\ref{eq:2}),~(\ref{eq:5})
	and~(\ref{eq:8}),
	by the reversibility of pass transitions, there exist stacks
	$S$ and $B$ such that
	the execution $\ex$ can be decomposed into an execution
	$\ex' : \Init(G_s) \to^* ((G_s, e_0), (\up, \square, S, B))$
	and at least one subsequent pass transition.
	In the execution $\ex'$, the link $e_0$ appears as a token
	position only once at the end, which can be checked in the
	same manner as the previous case~(\ref{item:PassDnUp}).
	Using this freshness of $e_0$ in $\ex'$ and
	Lem.~\ref{lem:Continuity}, we can conclude that any node that
	interacts with a token in the execution $\ex'$ (i.e.\ that is
	relevant in a pass transition in the execution $\ex'$) belongs
	to $E^\ddag$.
	This means that any pass transition in $\ex'$, on the starting
	graph $G_s$, can be imitated in the resulting graph $G_t$.
	Namely, the link $e_0$ corresponds to the result position
	$e_2$, and $\ex'$ corresponds to an execution
	$\ex'' : \Init(G_t) \to^* ((G_t, e_2), (\up, \square, S, B))$
	of only pass transitions, in which $e_2$ appears only once at
	the end.
	This execution $\ex''$ gives a desired execution that meets
	the condition (ii) of the binary relation $\preceq$.
  \item \label{item:EpsilonRwPass}
	In simulation of the basic rule~(\ref{eq:4}),
	the same reasoning as the previous case~(\ref{item:BetaRwEnd})
	gives an execution
	$\ex'' : \Init(G_t) \to^* ((G_t, e_0), (\up, \square, S, B))$
	of only pass transitions, in which $e_0$ appears only once at
	the end.
	By Lem.~\ref{lem:Continuity}, the result position $e_2$ does
	not appear in the execution $\ex''$.
	Therefore, $\ex''$ followed by pass transitions gives a
	desired execution that meets the condition (ii) of the binary
	relation $\preceq$.
  \item \label{item:SigmaRwEnd}
	In simulation of the basic rule~(\ref{eq:10}),
	by the reversibility of pass transitions, there exist an input
	$e_0$ of the $\lb{C}_{m+1}$-node and stacks
	$S$ and $B$ such that: $S' = S$, $B' = e_0:B$, and
	the execution $\ex$ can be decomposed into an execution
	$\ex' : \Init(G_s) \to^* ((G_s, e_0), (\up, \square, S, B))$
	and one subsequent pass transition that pushes $e_0$ to the
	box stack.
	By Lem.~\ref{lem:Continuity}, the link $e_3$ (see the figure)
	appears in the execution $\ex'$. Analysing this appearance, we
	can conclude that the link $e_0$ is in fact the main output of
	$(E')^\ddag_\emptyset$.
	\begin{itemize}
	 \item If $e_3$ appears with direction $\dn$ in $\ex'$,
	       because $e_3$ is an input of a function application
	       node or a $\lb{C}$-node, this appearance cannot be a
	       result of any pass transition. This is a contradiction.
	 \item If $e_3$ appears with direction $\up$,
	       it must be with rewrite flag $\square$, because $\ex'$
	       consists of pass transitions only.
	       Because $e_3$ is the main input of
	       $(E')^\ddag_\emptyset$, by Lem.~\ref{lem:PassEvalCtxt},
	       this appearance leads to a state
	       whose token position is the main output $e'$ of
	       $(E')^\ddag_\emptyset$, direction is $\up$ and rewrite
	       flag is $\square$.
	       One pass transition from the state leads to a state
	       whose token position is $e_1$.
	       This means there exists an execution $\ex'''$ of pass
	       transitions only, via the token position $e_3$ and the
	       second last token position $e'$,
	       to the token position $e_1$.
	       Because pass transitions are deterministic, it is
	       either:
	       (1) $\ex$ is strictly a sub-sequence of $\ex'''$,
	       (2) $\ex = \ex'''$, or
	       (3) $\ex'''$ is strictly a sub-sequence of $\ex$.
	       Because $\ex$ is followed by a pass transition and a
	       rewrite transition as shown in the figure, the case (1)
	       is impossible.
	       Because $e_1$ appears only once
	       at the end in the execution $\ex$, the case (3) leads
	       to a contradiction. Therefore we can conclude
	       that (2) is the case, i.e.\ $\ex = \ex'''$.
	       This means $e' = e_0$, i.e.\ $e_0$ is the main output
	       of $(E')^\ddag_\emptyset$.
	\end{itemize}
	As a consequence, the link $e_2$ is indeed the result
	position, corresponding to the link $e_0$.

	The rest of the reasoning is similar to the
	case~\ref{item:BetaRwEnd}. In the execution $\ex$ to the
	starting position $e_1$, the token does not interact with
	nodes that belong to $A^\dag$ or $v^\dag$; otherwise, by
	Lem.~\ref{lem:Continuity}, $e_1$ would have an extra
	appearance in $\ex$, which is a contradiction.
	For the same reason, the execution $\ex'$ to the link $e_0$
	does not involve any interaction of the token with the
	$\lb{C}_{m+1}$-node, and hence $e_0$ appears only once at the
	end in the execution $\ex'$.
	As a result, the execution $\ex'$ gives an execution
	$\ex'' : \Init(G_t) \to^* ((G_t, e_2), (\up, \square, S, B))$
	of only pass transitions on the resulting graph $G_t$, in
	which $e_2$ appears only once at the end.
	This execution $\ex''$ gives a desired execution that meets
	the condition (ii) of the binary relation $\preceq$. \qedhere
 \end{enumerate}
\end{proof}

\newcommand{\simsc}{0.65}

\begin{figure}[b]
 \begin{minipage}{\hsize}
  (\ref{eq:3}) \quad
  $\plug{E}{\W{\appLR{t}{u}}}
  \multimap_\epsilon \plug{E}{\appLR{\W{t}}{u}}$
 \end{minipage} \\
 \includegraphics[scale=\simsc]{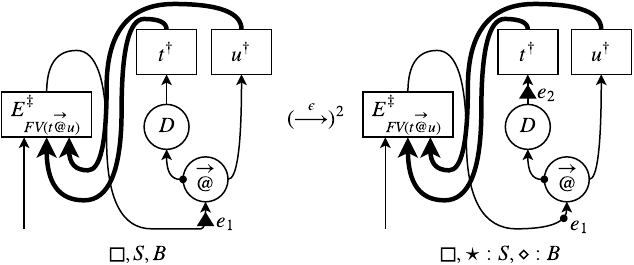} \\
 \begin{minipage}{\hsize}
  (\ref{eq:4}) \quad
  $\plug{E}{\appLR{\plug{A}{\W{\abs{x}{}{t}}}}{u}}
  \multimap_\epsilon \plug{E}{\appLR{\plug{A}{\abs{x}{}{t}}}{\W{u}}}$
 \end{minipage} \\
 \includegraphics[scale=\simsc]{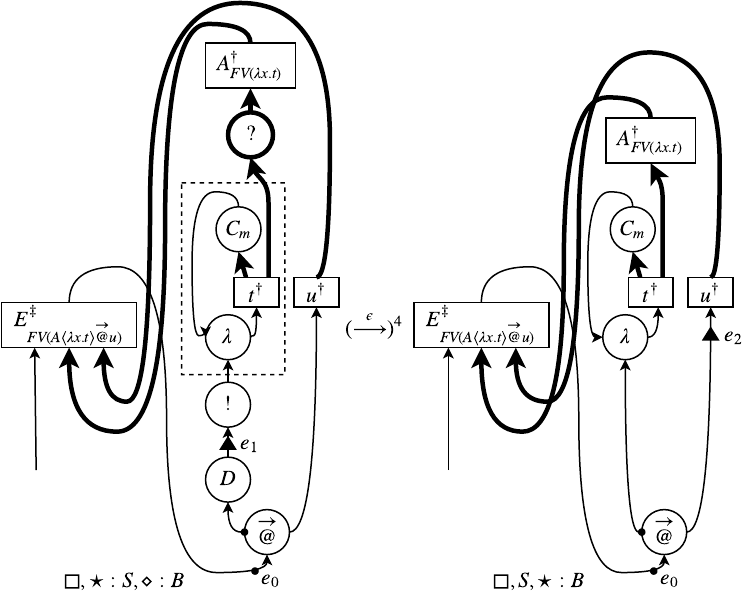} \\
 \begin{minipage}{\hsize}
  (\ref{eq:5}) \quad
  $\plug{E}{\appLR{\plug{A}{\abs{x}{}{t}}}{\plug{A'}{\W{v}}}}
  \multimap_\beta \plug{E}{\plug{A}{\W{t}[x \leftarrow \plug{A'}{v}]}}$
 \end{minipage} \\
 \includegraphics[scale=\simsc]{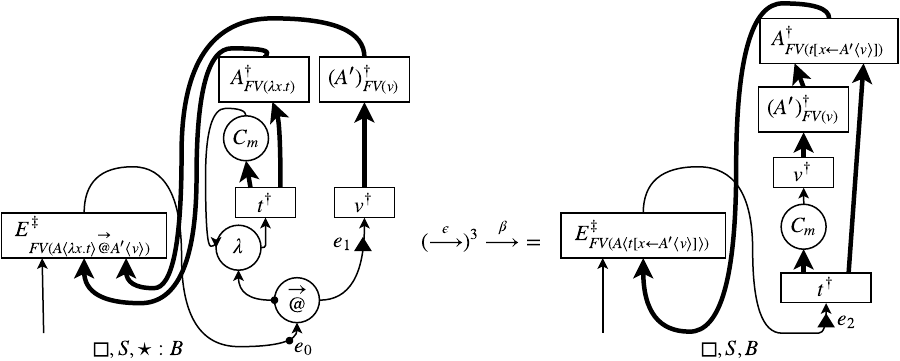}
 \caption{Illustration of Simulation: Left-to-Right Call-by-Value Application}
 \label{fig:Proof2}
\end{figure}

\begin{figure}[p]
 \begin{minipage}{\hsize}
  (\ref{eq:1}) \quad
  $\plug{E}{\W{\appLazy{t}{u}}}
  \multimap_\epsilon \plug{E}{\appLazy{\W{t}}{u}}$
 \end{minipage} \\
 \includegraphics[scale=\simsc]{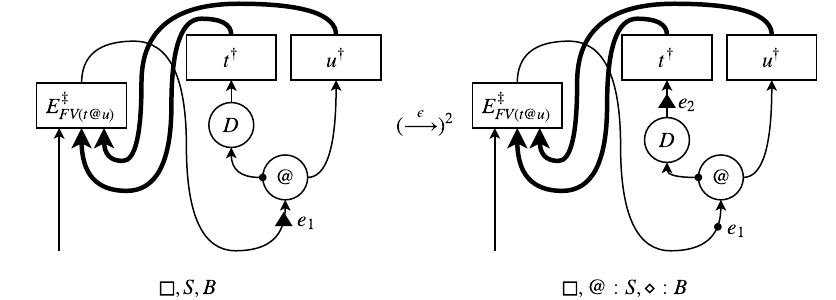} \\
 \begin{minipage}{\hsize}
  (\ref{eq:2}) \quad
  $\plug{E}{\appLazy{\plug{A}{\W{\abs{x}{}{t}}}}{u}}
  \multimap_\beta \plug{E}{\plug{A}{\W{t}[x \leftarrow u]}}$
 \end{minipage} \\
 \includegraphics[scale=\simsc]{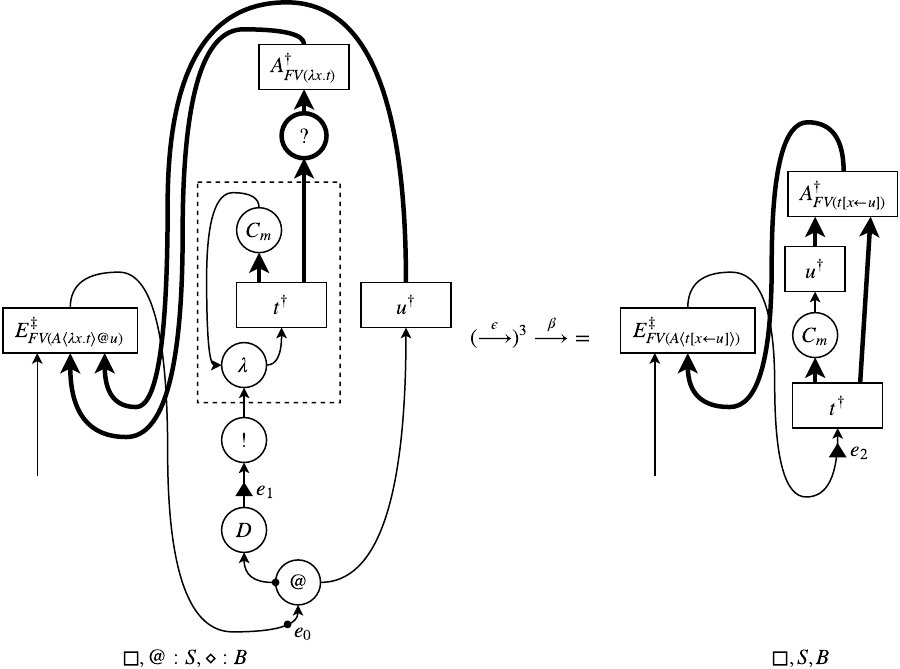} \\
 \begin{minipage}{\hsize}
  (\ref{eq:9}) \quad
  $\plug{E}{\plug{E'}{\W{x}}[x \leftarrow \plug{A}{u}]}
  \multimap_\epsilon
  \plug{E}{\plug{E'}{x}[x \leftarrow \plug{A}{\W{u}}]}$
 \end{minipage} \\
 \includegraphics[scale=\simsc]{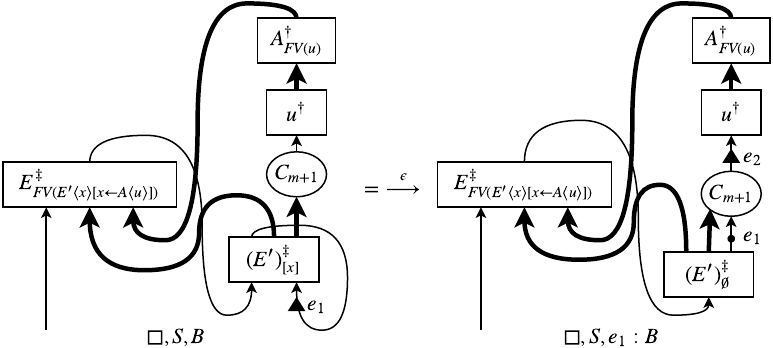} \\
 \begin{minipage}{\hsize}
  (\ref{eq:10}) \quad
  $\plug{E}{\plug{E'}{x}[x \leftarrow \plug{A}{\W{v}}]}
  \multimap_\sigma
  \plug{E}{\plug{A}{\plug{E'}{\W{v}}[x \leftarrow v]}}$
 \end{minipage} \\
 \includegraphics[scale=\simsc]{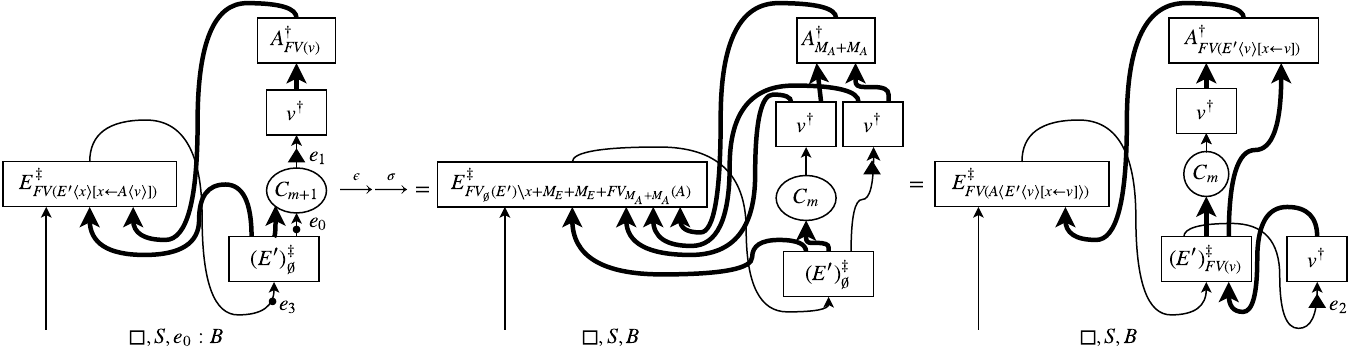}
  \caption{Illustration of Simulation: Call-by-Need Application and
 Explicit Substitutions}
 \label{fig:Proof1}
\end{figure}

\begin{figure}[t]
 \begin{minipage}{\hsize}
  (\ref{eq:6}) \quad
  $\plug{E}{\W{\appRL{t}{u}}}
  \multimap_\epsilon \plug{E}{\appRL{t}{\W{u}}}$
 \end{minipage} \\
 \includegraphics[scale=\simsc]{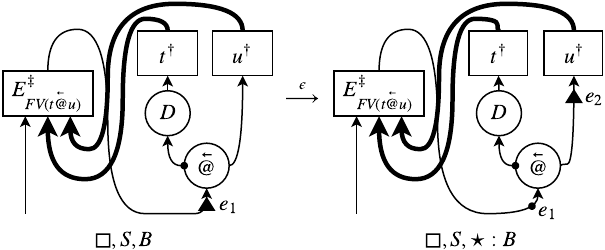} \\
 \begin{minipage}{\hsize}
  (\ref{eq:7}) \quad
  $\plug{E}{\appRL{t}{\plug{A}{\W{v}}}}
  \multimap_\epsilon \plug{E}{\appRL{\W{t}}{\plug{A}{v}}}$
 \end{minipage} \\
 \includegraphics[scale=\simsc]{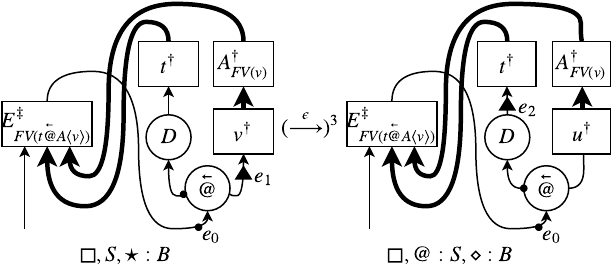} \\
 \begin{minipage}{\hsize}
  (\ref{eq:8}) \quad
  $\plug{E}{\appRL{\plug{A}{\W{\abs{x}{}{t}}}}{\plug{A'}{v}}}
  \multimap_\beta \plug{E}{\plug{A}{\W{t}[x \leftarrow \plug{A'}{v}]}}$
 \end{minipage} \\
 \includegraphics[scale=\simsc]{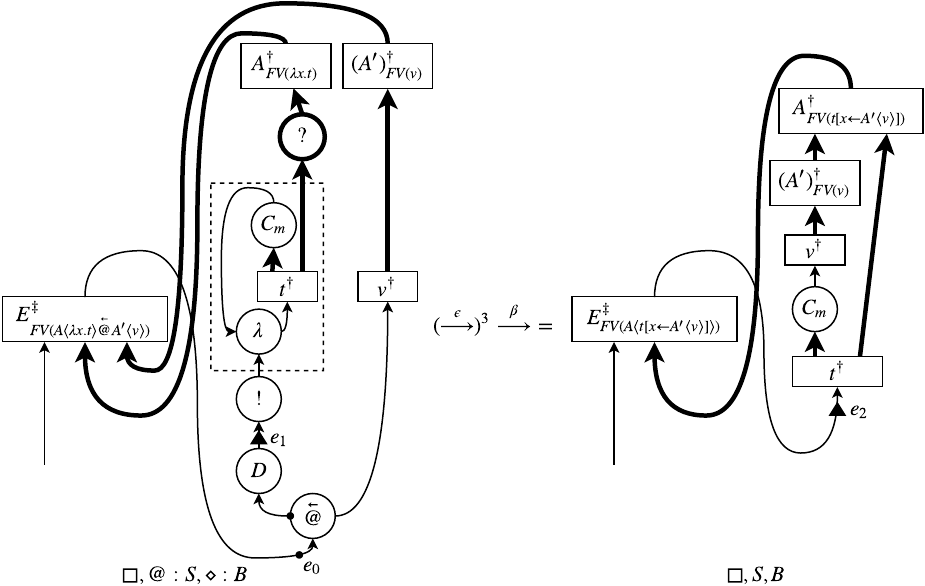}
 \caption{Illustration of Simulation: Right-to-Left Call-by-Value Application}
 \label{fig:Proof3}
\end{figure}


\section{Time-Cost Analysis}
\label{sec:time-cost-analysis}

We analyse how time-efficiently the token-guided graph-rewriting
machine implements evaluation strategies, following the methodology
developed by Accattoli et al.\
\cite{AccattoliBM14,AccattoliSC14,Accattoli16}.
The time-cost analysis focuses on how efficiently an abstract machine
implements an evaluation strategy.
In other words, we are not interested in minimising the number of
$\beta$-reduction steps simulated by an abstract machine.
Our aim is to see if the number of transitions of an
abstract machine is ``reasonable'', compared to the number of
necessary $\beta$-reduction steps determined by a given evaluation
strategy.

Accattoli's methodology assumes that an abstract machine has three
groups of transitions: 1) ``$\beta$-transitions'' that correspond to
$\beta$-reduction in which substitution is delayed, 2) transitions
that perform substitution, and 3) other ``overhead'' transitions.
We incorporate this classification using the labels
$\beta$, $\sigma$ and $\epsilon$ of transitions.

Another assumption of the methodology is that, each step of
$\beta$-reduction is simulated by a single transition of an abstract
machine, and so is substitution of each occurrence of a variable.
This is satisfied by many known abstract machines, including
Danvy and Zerny's storeless abstract machine~\cite{DanvyZ13} that our
sub-machine semantics resembles, however not by the token-guided
graph-rewriting abstract machine.
The machine has ``finer'' transitions and can take
several transitions to simulate a single step of reduction, as we can
observe in Thm.~\ref{thm:WeakSimulation}.
In spite of this mismatch we can still follow the methodology, thanks
to the weak simulation $\preceq$.
It discloses what transitions of the token-guided graph-rewriting
machine exactly correspond to
$\beta$-reduction and substitution, and gives a concrete number of
overhead transitions that the machine needs to simulate
$\beta$-reduction and substitution.

The methodology of time-cost analysis has four steps:
(I) bound the number of transitions required in implementing
evaluation strategies, (II) estimate time cost of each transition,
(III) bound overall time cost of implementing evaluation strategies,
by multiplying the number of transitions with time cost for each
transition, and finally (IV) classify the abstract machine according
to its execution time cost.
Consider now the following taxonomy of abstract
machines introduced in~\cite{Accattoli16}.
\begin{defi}
 [classes of abstract machines {\cite[Def.~7.1]{Accattoli16}}]
 \label{def:taxonomy}
 \noindent
 \begin{enumerate}
  \item An abstract machine is \emph{efficient} if its execution time
	cost is linear in both the input size and the number of
	$\beta$-transitions.
  \item An abstract machine is \emph{reasonable} if its execution time
	cost is polynomial in the input size and the number of
	$\beta$-transitions.
  \item An abstract machine is \emph{unreasonable} if it is not
	reasonable.
 \end{enumerate}
\end{defi}

In our case, the input size is given by 
the \emph{size} $|t|$ of the term $t$, inductively defined by:
\begin{align*}
 |x| &:= 1,
 & |\abs{x}{}{t}| &:= |t| + 1, \\
 |\appLazy{t}{u}| = |\appLR{t}{u}| = |\appRL{t}{u}| &:= |t| + |u| + 1,
 & |t[x \leftarrow u]| &:= |t| + |u| + 1.
\end{align*}
The number of $\beta$-transitions is simply the number of transitions
labelled with $\beta$, which in fact corresponds to the number of
reductions labelled with $\beta$, thanks to
Thm.~\ref{thm:WeakSimulation}.

Given an evaluation $\ev$, the number of occurrences of a label $\chi$
is denoted by $|\ev|_\chi$.
The sub-machine semantics comes with the following quantitative
bounds.
\begin{prop}
 \label{prop:MockMachineBounds}
 For any pure closed term $t$ and
 any evaluation $\ev \colon \W{t} \multimap^* \plug{A}{\W{v}}$ that
 terminates, the number of reductions is bounded by
 $
  |\ev|_\sigma = \mathcal{O}(|\ev|_\beta)$ and $
  |\ev|_\epsilon = \mathcal{O}(|t|\cdot|\ev|_\beta).
 $
\end{prop}
\begin{proof}
 A term uses a single evaluation strategy,
   either call-by-need, left-to-right call-by-value, or
 right-to-left call-by-value.
 Forgetting the window of an enriched term
 $\plug{E}{\W{t}}$ gives a term $\plug{E}{t}$, which can be seen as a
 term of the linear substitution calculus~\cite{AccattoliK10}.
 This gives an one-to-one correspondence between an
 evaluation by the sub-machine semantics and a ``derivation'' in the
 linear substitution calculus, via the concept of
 ``distillery''~\cite[Sec.~4]{AccattoliBM14}.
 The correspondence is in such a way that it enables us to directly
 apply the bounds about the linear substitution
 calculus~\cite[Cor.~1 \& Thm.~2]{AccattoliSC14} and obtain the first
 equation.

 The second equation is proved by combining the first equation and
 an equation
 $|\ev|_\epsilon = \mathcal{O}(|t|\cdot(|\ev|_\beta + |\ev|_\sigma))$.
 This auxiliary equation can be proved using ideas from
 Accattoli et al.'s analysis of various abstract
 machines~\cite[Thm.~11.3 \& Thm.~11.5]{AccattoliBM14}, as below.


 For any enriched term $\plug{E'}{\W{t'}}$ that appears in the
 evaluation
 $\ev \colon \W{t} \multimap^* \plug{A}{\W{v}}$, we define two
 measures.
 The first measure $\#_1(\plug{E'}{\W{t'}})$ is defined by:
 $|t'| + |u|$ if $E'$ is in the form of
 $\plug{E''}{\appLazy{\plug{A'}{\cdot}}{u}}$,
 $\plug{E''}{\appLR{\plug{A'}{\cdot}}{u}}$, or
 $\plug{E''}{\appRL{u}{\plug{A'}{\cdot}}}$;
 and $|t'|$ otherwise.
 By Lem.~\ref{lem:SubMachineInvariants}, both $t'$ and $u$ above are
 sub-terms of $t$, and we have
 $\#_1(\plug{E'}{\W{t'}}) \leq 2 \cdot |t|$.
 The second measure $\#_2(E')$ is on $E'$ only, and defined
 inductively as below.
 \begin{align*}
  \#_2(\emptyCtxt) := 0, \quad
  \#_2(\plug{E''}{x}[x \leftarrow E'''])
  &:= \#_2(E'') + \#_2(E''') + 1, \\
  \#_2(\appLazy{E''}{t}) = \#_2(\appLR{E''}{t})
  = \#_2(\appLR{\plug{A}{v}}{E''})
  &:= \#_2(E''), \\
  \#_2(\appRL{t}{E''})
  = \#_2(\appRL{E''}{\plug{A}{v}})
  = \#_2(E''[x \leftarrow t''])
  &:= \#_2(E'').
 \end{align*}

 Because the basic
 rules~(\ref{eq:1}),~(\ref{eq:3}),~(\ref{eq:4}),~(\ref{eq:6})
 and~(\ref{eq:7}) strictly reduce the measure $\#_1$,
 these rules can be consecutively applied at most
 $2 \cdot |t|$ times.
 The evaluation $\ev$ can be seen as applications of these rules
 interleaved with other rules, so the total number of applications of
 these five basic rules can be bounded by
 $\mathcal{O}(|t| \cdot (|\ev|_\beta + |\ev|_\sigma + |\ev|_9))$,
 where $|\ev|_9$ denotes the total number of applications of the basic
 rule~(\ref{eq:9}).

 The measure $\#_2$ is increased only by the basic rule~(\ref{eq:9})
 and decreased only by the basic rule~(\ref{eq:10}).
 Both the increase and the decrease are of one.
 Because the measure $\#_2$ gives zero for both
 $\W{t}$ and $\plug{A}{\W{v}}$, namely
 $\#_2(\emptyCtxt) = \#_2(A) = 0$, the basic rule~(\ref{eq:9})
 must be applied as many times as the basic rule~(\ref{eq:10}) in the
 evaluation $\ev$.
 This means $|\ev|_\sigma = |\ev|_9$.
 
 Combining the bound
 $\mathcal{O}(|t| \cdot (|\ev|_\beta + |\ev|_\sigma + |\ev|_9))$ with
 the equation $|\ev|_\sigma = |\ev|_9$ gives
 the auxiliary equation on $|\ev|_\epsilon$.
\end{proof}

We use the same notation $|\ex|_\chi$, as for an evaluation, to denote
the number of occurrences of each label $\chi$ in an execution $\ex$.
Additionally the number of rewrite transitions with the label
$\epsilon$, i.e.\ those that eliminates a $\oc$-box structure, is
denoted by $|\ex|_\mathit{\epsilon R}$.
Note that pass transitions are all labelled with $\epsilon$, and
hence $|\ex|_\mathit{\epsilon R} \leq |\ex|_\epsilon$.
The following proposition completes the first step of
the cost analysis.
\begin{prop}[Soundness \& completeness, with number bounds]
 \label{prop:GraphRewriterBounds}
 For any pure closed term $t$,
 an evaluation $\ev \colon \W{t} \multimap^* \plug{A}{\W{v}}$
 terminates with the enriched term $\plug{A}{\W{v}}$ if and only if an
 execution
 $\ex \colon
 \Init(t^\dag) \to^* \Final(A^\ddag \circ v^\dag)$
 terminates with the graph $A^\ddag \circ v^\dag$.
 Moreover the number of transitions is bounded by
$  |\ex|_\beta = |\ev|_\beta$,
$  |\ex|_\sigma = \mathcal{O}(|\ev|_\beta)$,
$  |\ex|_\epsilon = \mathcal{O}(|t|\cdot|\ev|_\beta)$,
$  |\ex|_\mathit{\epsilon R} = \mathcal{O}(|\ev|_\beta)$.
\end{prop}
\begin{proof}
 %
 Because the initial term $t$ is closed,
 any enriched term $\plug{E'}{\W{t'}}$ that appears in the
 evaluation $\ev$ is also closed.
 This implies that a reduction is always possible at
 $\plug{E'}{\W{t'}}$ unless it is in the form of $\plug{A'}{v'}$.
 In particular, if $t'$ is a variable, the
 variable is captured by an explicit substitution in $E'$ and the
 basic rule~(\ref{eq:10}) is possible.
 Consequently, if an evaluation of the pure closed term $t$
 terminates, the last enriched term is in the form of $\plug{A'}{v'}$.

 The forward direction of the equivalence, that is, the evaluation
 $\ev$ implies the execution $\ex$, follows from
 Thm.~\ref{thm:WeakSimulation}. The backward direction, that is, the
 execution $\ex$ implies the evaluation $\ev$, also follows from
 Thm.~\ref{thm:WeakSimulation}, because an evaluation of the pure
 closed term $t$ is in the form of $\W{t} \multimap^* \plug{A}{\W{v}}$
 or never terminates.

 Thm.~\ref{thm:WeakSimulation} also gives equations
 $|\ex|_\beta = |\ev|_\beta$, $|\ex|_\sigma = |\ev|_\sigma$ and
 $|\ex|_\epsilon =
 \mathcal{O}(|\ev|_\beta + |\ev|_\sigma + |\ev|_\epsilon)$.
 Combining these with Prop.~\ref{prop:MockMachineBounds} yields the
 desired equations except for the last one
 (i.e.\ $|\ex|_\mathit{\epsilon R} = \mathcal{O}(|\ev|_\beta)$).

 This last equation follows from an equation
 $|\ex|_\mathit{\epsilon R} = |\ex|_\beta$ that can be proved as
 follows.
 For any graph state $((G,e),\delta)$ that appears in the execution
 $\ex \colon
 \Init(t^\dag) \to^* \Final(A^\ddag \circ v^\dag)$,
 we define a measure $\#(G)$ by the number of $\lambda$-nodes that are
 outside any $\oc$-box in the graph $G$.

 Firstly, at any point of the execution $\ex$, the token is inside a
 $\oc$-box if and only if it has the rewrite flag `$\oc$'. This means,
 if a $\lambda$-node gets eliminated by a rewrite transition labelled
 with $\beta$, the $\lambda$-node is outside a $\oc$-box.
 By Lem.~\ref{lem:SubGraph}, each $\oc$-box has exactly one
 $\lambda$-node that directly belongs to it.
 It follows that
 each rewrite transition labelled with $\epsilon$ brings exactly one
 $\lambda$-node outside a $\oc$-box.

 As a result,
 each rewrite transition labelled with $\beta$
 decreases the measure $\#$ by one, and
 each rewrite transition labelled with $\epsilon$
 increases the measure $\#$ by one. No other transitions change the
 measure $\#$.
 Because the measure $\#$ gives zero for the initial and final graph
 states $\Init(t^\dag)$ and $\Final(A^\ddag \circ v^\dag)$,
 namely $\#(t^\dag) = \#(A^\ddag \circ v^\dag) = 0$,
 we have
 $|\ex|_\mathit{\epsilon R} = |\ex|_\beta$.
\end{proof}

The next step in the cost analysis is to estimate the time cost of
each transition.
We assume that graphs are implemented in the following way.
Each $\wn$-node, and its input and output, are identified and
implemented as a single link.
Each link is given by two pointers to its child and its parent.
If a node is not a $\wn$-node, it is given by its label, pointers to
its inputs, and pointers to its outputs; the pointers to inputs are
omitted for $\lb{C}$-nodes.
Additionally, each link and node has a pointer to a $\oc$-node, or a
null pointer, to indicate the $\oc$-box structure it directly belongs
in.
Note that each link has at most three pointers, and
each node has at most two input (resp.\ output) pointers,
which are distinguished.
The \emph{size} of a graph can be estimated using the number of nodes
that are not $\wn$-nodes.
Accordingly, a position of the token is a pointer to a link, a
direction and a rewrite flag are two symbols, a computation stack is a
stack of symbols, and finally a box stack is a stack of symbols and
pointers to links.

Using these assumptions of implementation, we estimate time cost of
each transition.
All pass transitions have constant cost.
Each pass transition looks up one node and its outputs (that
are either one or two) next to the current position, and involves
a fixed number of elements of the token data.
Rewrite transitions with the label $\beta$ have constant cost,
as they change a constant number of nodes and links, and only a
rewrite flag of the token data.
Rewrite transitions with the label $\epsilon$ remove a
$\oc$-box structure. This can be done by traversing nodes from its
principal door, and hence have cost bounded by the size of the
$\oc$-box.
Finally, rewrite transitions with the label $\sigma$ copy a $\oc$-box
structure.
Copying cost is bounded by the size of the $\oc$-box.
Updating cost of the sub-graph $H'$ (see
Fig.~\ref{fig:RewriteTransitions}) is bounded by the number of
auxiliary doors, which is less than the size of the copied
$\oc$-box.
Updating cost of the $\lb{C}$-node is constant, because
$\lb{C}$-nodes do not have pointers to its inputs,
by the assumption about the implementation of graphs.

With the results of the previous two steps, we can now give the
overall time cost of executions and classify our abstract machine.
\begin{thm}[Soundness \& completeness, with cost bounds]
 \label{thm:CostBounds}
 For any pure closed term $t$,
 an evaluation $\ev \colon \W{t} \multimap^* \plug{A}{\W{v}}$
 terminates with the enriched term $\plug{A}{\W{v}}$ if and only if an
 execution
 $\ex \colon
 \Init(t^\dag) \to^* \Final(A^\ddag \circ v^\dag)$
 terminates with the graph $A^\ddag \circ v^\dag$.
 The overall time cost of the execution $\ex$ is bounded by
 $\mathcal{O}(|t|\cdot|\ev|_\beta)$.
\end{thm}
\begin{proof}
 Non-constant cost of rewrite transitions is the size of a $\oc$-box.
 By Lem.~\ref{lem:SubGraph}, this size is less than the size of the
 initial graph $t^\dag$, which can be bounded by the size
 $|t|$ of the initial term.
 Therefore any non-constant cost of each rewrite transition, in the
 execution $\ex$, can be also bounded by $|t|$.
 By Prop.~\ref{prop:GraphRewriterBounds},
 the overall time cost of rewrite transitions labelled with $\beta$ is
 $\mathcal{O}(|\ev|_\beta)$, and that of the other rewrite transitions
 and pass transitions is $\mathcal{O}(|t|\cdot|\ev|_\beta)$.
\end{proof}

Note that the time cost of constructing the initial graph $t^\dag$,
and attaching a token to it, does not affect the bound
$\mathcal{O}(|t|\cdot|\ev|_\beta)$, because it can be done in linear
time with respect to $|t|$. This is thanks to the assumption about
implementation, namely that $\wn$-nodes and input pointers of
$\lb{C}$-nodes are omitted.

\begin{cor}
 \label{cor:classify}
 The token-guided graph-rewriting machine is an efficient abstract
 machine implementing call-by-need, left-to-right call-by-value and
 right-to-left call-by-value evaluation strategies, in the sense of
 Def.~\ref{def:taxonomy}.
\end{cor}

Cor.~\ref{cor:classify} classifies the graph-rewriting machine as not
just ``reasonable'', but in fact ``efficient''.
In terms of token passing, this efficiency benefits from the graphical
representation of
environments (i.e.\ explicit substitutions in our setting).
The graphical representation is in such a way that each bound
variable is associated with exactly one $\lb{C}$-node, which is
ensured by the translations $(\cdot)^\dag$ and $(\cdot)^\ddag$ and the
rewrite transition $\to_\sigma$.
Excluding any two sequentially-connected $\lb{C}$-nodes is essential
to achieve the ``efficient'' classification, because it yields the
constant cost to look up a bound variable and its associated
computation.

As for graph rewriting, the ``efficient'' classification shows that
introduction of graph rewriting to token passing does not bring in any
inefficiencies.
In our setting, graph rewriting brings in two kinds of non-constant
cost. One is duplication cost of a sub-graph, which is indicated by a
$\oc$-box, and the other is elimination cost of a $\oc$-box that
delimits abstraction.
Unlike the duplication cost, the elimination cost leads to non-trivial
cost that abstract machines in the literature usually do not have.
Namely, our graph-rewriting machine simulates a
$\beta$-reduction step, in which an abstraction constructor is
eliminated and substitution is delayed, at the non-constant cost
depending on the size of the abstraction.
The time-cost analysis confirms that the duplication cost and
the unusual elimination cost have the same impact, on the overall
time cost, as the cost of token passing. What is vital here is the
sub-graph property (Lem.~\ref{lem:SubGraph}), which ensures that the
cost of each duplication and elimination of a $\oc$-box is always
linear in the input size.

\section{Rewriting vs. Jumping}
\label{sec:rewr-vs.-jump}

The starting point of our development is the GoI-style token-passing
abstract machines for call-by-name evaluation, given by Danos and
Regnier~\cite{DanosR96}, and by Mackie~\cite{Mackie95}.
Fig.~\ref{fig:IAM} recalls these token-passing machines as a version
of the DGoIM with the passes-only interleaving strategy (i.e.\ the
DGoIM with only pass transitions).
It follows the convention of Fig.~\ref{fig:PassTransitions}, but
a black triangle in the figure points along (resp.\ against) the
direction of the edge if the token direction is $\up$ (resp.\ $\dn$).
Note that this version uses different token data, to which we will
come back later.

\begin{figure}[t]
 \centering
 \begin{minipage}{.9\linewidth}
  Token data $(d,S,B,E)$ consists of:
  \begin{itemize}
   \item a \emph{direction} defined by
	 $d ::= \up \mid \dn$,
   \item a \emph{computation stack} defined by
	 $S ::= \square \mid \mathsf{A}:S \mid @:S$, and
   \item a \emph{box stack} $B$ and an \emph{environment stack} $E$,
	 both defined by
	 $B,E ::= \square \mid \sigma:B$,
	 using \emph{exponential signatures}
	 $\sigma ::= \star \mid e\cdot\sigma
	 \mid \langle\sigma,\sigma\rangle$
	 where $e$ is any link of the underlying graph.
  \end{itemize}
  Pass transitions:
  \begin{center}
  \includegraphics[scale=0.85]{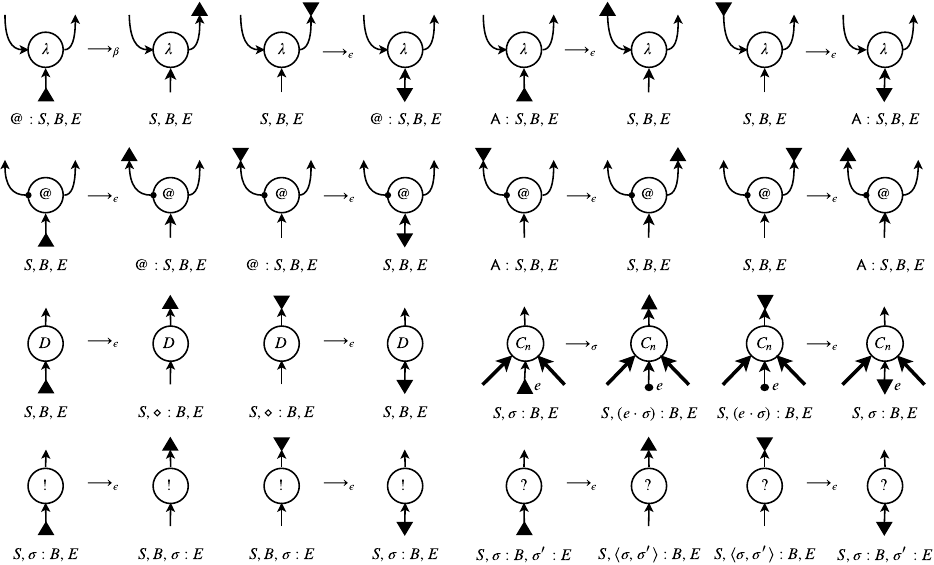}
  \end{center}
  Given a term $t$ with the call-by-need function application ($@$)
  abused, a successful execution
  $((t^\dag,e_t),(\up,\square,\square,\square,\square))\to^*
  ((t^\dag,e_v),(\up,\square,\square,\square,\square))$
  starts at the root $e_t$ of the translation $t^\dag$, and ends at
  the root $e_v$ of the translation $v^\dag$, for some sub-value $v$
  of the term $t$. The value $v$ indicates the evaluation result.
 \end{minipage}
 \caption{Passes-Only DGoIM for Call-by-Name~\cite{DanosR96,Mackie95}}
 \label{fig:IAM}
\end{figure}

Token-passing GoI keeps the underlying graph fixed, and re-evaluates a
term by repeating token moves. It therefore favours space efficiency
at the cost of time efficiency.
Repeating token actions poses a challenge for evaluations in which
duplicated computation must not lead to repeated evaluation,
especially call-by-value
evaluation~\cite{FernandezM02,Schoepp14b,HoshinoMH14,DalLagoFVY15}.
Moreover, in call-by-value repeating token actions raises the
additional technical challenge of avoiding repeating any associated
computational effects~\cite{Schoepp11,MuroyaHH16,DalLagoFVY17}.
A partial solution to this conundrum is to focus on the soundness of
the equational theory, while deliberately ignoring the time
costs~\cite{MuroyaHH16}.
Introduction of graph reduction, the key idea of the DGoIM, is one
total solution in the sense that it avoids repeated token moves and
also improves time efficiency of token-passing GoI.
Another such solution in the literature is introduction of jumps.
We discuss how these two solutions affect machine design and
space efficiency.

The most greedy way of introducing graph reduction, namely the
rewrites-first interleaving we studied in this work, simplifies
machine design in terms of the variety of pass transitions and token
data.
First, some token moves turn irrelevant to an execution.
This is why Fig.~\ref{fig:PassTransitions} for the rewrites-first
interleaving has fewer pass transitions than Fig.~\ref{fig:IAM} for
the passes-only interleaving.
Certain nodes, like `$\wn$', always get eliminated before visited by
the token, in the rewrites-first interleaving.
Accordingly, token data can be simplified.
The box stack and the environment stack used in Fig.~\ref{fig:IAM} are
integrated to the single box stack used in
Fig.~\ref{fig:PassTransitions}.
The integrated stack does not need to carry the exponential
signatures.
They make sure that the token exits $\oc$-boxes
appropriately in the token-passing GoI, by maintaining binary tree
structures, but the token never exits $\oc$-boxes with the
rewrites-first interleaving.
Although the rewrites-first interleaving simplifies token data,
rewriting itself, especially duplication of sub-graphs, becomes the
source of space-inefficiency.

A jumping mechanism can be added on top of the token-passing GoI, and
enables the token to jump along the path it would otherwise
follow step-by-step.
Although no quantitative analysis is provided, it gives time-efficient
implementations of evaluation strategies, namely of call-by-name
evaluation~\cite{DanosR96} and call-by-value
evaluation~\cite{FernandezM02}.
Jumping can reduce the variety of pass transitions, like rewriting, by
letting some nodes always be jumped over. Making a jump is just
changing the token position, so jumping can be described as a
variation of pass transitions, unlike rewriting.
However, introduction of jumping rather complicates token data.
Namely it requires partial duplications of token data, which not only
complicates machine design but also damages space efficiency.
The duplications effectively represent virtual copies of sub-graphs,
and accumulate during an execution.
Tracking virtual copies is the trade-off of keeping the underlying
graph fixed.
Some jumps that do not involve virtual copies can be described as a
form of graph rewriting that eliminates nodes.

\begin{figure}[t]
 \centering
 \begin{minipage}{.9\linewidth}
  Token data $(d,S,B,E)$ consists of:
  \begin{itemize}
   \item a \emph{direction} defined by
	 $d ::= \up \mid \dn$,
   \item a \emph{computation stack} defined by
	 $S ::= \square \mid \mathsf{A}:S \mid @:S$, and
   \item a \emph{box stack} $B$ and an \emph{environment stack} $E$,
	 both defined by
	 $B,E ::= \square \mid (e,E):B$,
	 where $e$ is any link of the underlying graph.
  \end{itemize}
  Pass transitions:
  \begin{center}
   \includegraphics[scale=0.85]{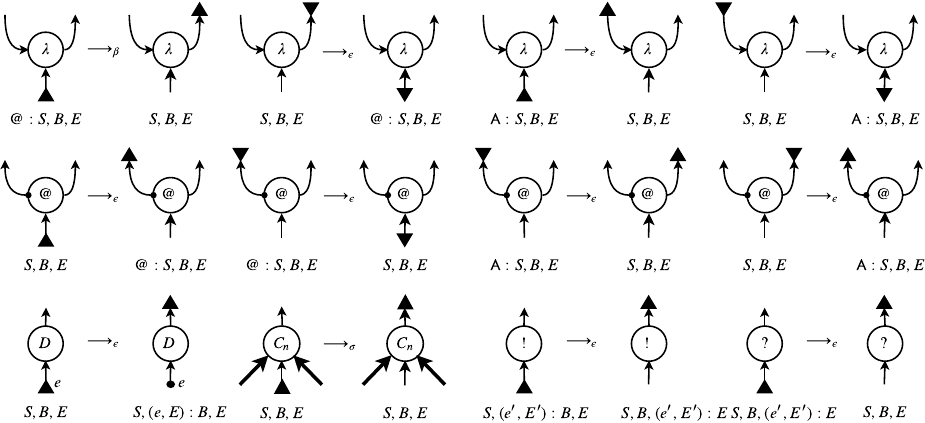}
  \end{center}
  Jump transition:
  $((G,e),(\dn,S,B,(e',E'):E)) \to_\epsilon ((G,e'),(\dn,S,B,E'))$,
  where the old position $e$ is the output of a $\oc$-node:
  \parbox[c]{2em}{\vspace{1ex}\includegraphics[scale=0.85]{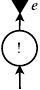}\vspace{1ex}}. \\
  Given a term $t$ with the call-by-need function application ($@$)
  abused, a successful execution
  $((t^\dag,e_t),(\up,\square,\square,\square,\square))\to^*
  ((t^\dag,e_v),(\up,\square,\square,\square,\square))$
  starts at the root $e_t$ of the translation $t^\dag$, and ends at
  the root $e_v$ of the translation $v^\dag$, for some sub-value $v$
  of the term $t$. The value $v$ indicates the evaluation result.
 \end{minipage}
 \caption{Passes-Only DGoIM plus Jumping for
 Call-by-Name~\cite{DanosR96}}
 \label{fig:JAM}
\end{figure}

Finally, we give a quantitative comparison of space usage between
rewriting and jumping.
As a case study, we focus on implementations of call-by-name/need
evaluation, namely on
the passes-only DGoIM recalled in Fig.~\ref{fig:IAM}, our
rewrites-first DGoIM, and the passes-only DGoIM equipped with jumping
that we will recall in Fig.~\ref{fig:JAM}.
A similar comparison is possible for left-to-right call-by-value
evaluation, between our rewrites-first DGoIM and the jumping
machine given by Fern{\'{a}}ndez and Mackie~\cite{FernandezM02}.

Fig.~\ref{fig:JAM} recalls Danos and Regnier's token-passing machine
equipped with jumping~\cite{DanosR96}, which is proved to be
isomorphic to Krivine's abstract machine~\cite{Krivine07} for
call-by-name evaluation.
The machine has pass transitions as well as the \emph{jump} transition
that lets the token jump to a remote position\footnote{Our on-line
visualiser additionally supports this jumping machine.}.
Compared with the token-passing GoI (Fig.~\ref{fig:IAM}),
pass transitions for nodes related to $\oc$-boxes are reduced and
changed, so that the jumping mechanism imitates rewrites involving
$\oc$-boxes. The token remembers its old position, together with its
current environment stack, when passing a $\lb{D}$-node upwards. The
token uses this information and make a jump back in the jump
transition, in which the token exits a $\oc$-box at
the principal door ($\oc$-node) and changes its position to the
remembered link $e'$.

The quantitative comparison, whose result is stated below, shows
partial duplication of token data impacts space usage much more than
duplication of sub-graphs, and therefore rewriting has asymptotically
better space usage than jumping.
\begin{prop}
 After $n$ transitions from an initial state of a graph of size
 $|G_0|$, space usage of three versions of the DGoIM is bounded as in
 the table below.
 \begin{table}[h]
  \begin{tabular}{|c||c||c|c|}
   \hline
   machines & token-passing only & rewriting added & jumping added \\
   & (Fig.~\ref{fig:IAM})
   & (Fig.~\ref{fig:PassTransitions}
   \& Fig.~\ref{fig:RewriteTransitions})
   & (Fig.~\ref{fig:JAM}) \\ \hline
   evaluations implemented
   & call-by-name & call-by-need & call-by-name \\ \hline \hline
   size of graph
   & $|G_0|$ & $\mathcal{O}(n\cdot |G_0|)$ & $|G_0|$ \\ \hline
   size of token position
   & $\log{|G_0|}$ & $\mathcal{O}{(\log{(n\cdot |G_0|)})}$
   & $\log{|G_0|}$ \\ \hline
   size of token data
   & $\mathcal{O}(n\cdot\log{|G_0|})$
   & $\mathcal{O}(n\cdot\log{(n\cdot |G_0|)})$
   & $\mathcal{O}(2^n\cdot\log{|G_0|})$ \\ \hline
  \end{tabular}
 \end{table}
\end{prop}
\begin{proof}
 The size $|G_n|$ of the underlying graph after $n$ transitions can be
 estimated using the size $|G_0|$ of the initial graph.
 Our rewrites-first DGoIM is the only one that changes the underlying
 graph during an execution. Thanks to the sub-graph property
 (Lem.~\ref{lem:SubGraph}), the size $|G_n|$ can be bounded as
 $|G_n| = \mathcal{O}(n_\sigma \cdot |G_0|)$,
 where $n_\sigma$ is the number of $\sigma$-labelled transitions in the
 $n$ transitions.
 In the token-passing machines with and without jumping
 (Fig.~\ref{fig:IAM} and Fig.~\ref{fig:JAM}), clearly $|G_n| = |G_0|$.
 In any of the three machines, the token position can be represented in
 the size of $\log{|G_n|}$.

 Next estimation is of token data. Because stacks can have a link of the
 underlying graph as an element, the size of token data after $n$
 transitions depends on $\log{|G_n|}$.
 Both in the token-passing machine (Fig.~\ref{fig:IAM}) and our
 rewrites-first DGoIM, at most one element is pushed in each
 transition.
 Therefore the size of token data is bounded by
 $n\cdot\mathcal{O}(\log{(|G_n|)})$.
 On the other hand, in the jumping machine (Fig.~\ref{fig:JAM}), the
 size of token data, especially the box stack and the environment
 stack, can grow exponentially because of the partial duplication.
 Therefore token data has the size
 $\mathcal{O}(2^n\cdot\log{(|G_n|)})$.
 For example, a term with many $\eta$-expansions, like
 $\app{(\abs{f}{}{
 (\abs{x}{}{\app{(\abs{y}{}{\app{(\abs{z}{}{\app{f}{z}})}{y}})}{x}})
 })}
 {(\abs{w}{}{w})}$,
 causes exponential grow of the box stack in the jumping machine.
\end{proof}

\section{Related Work and Conclusion}
\label{sec:relat-work-concl}

In an abstract machine of any functional programming language,
computations assigned to variables have to be stored for later use.
Potentially multiple, conflicting, computations can be assigned to a
single variable, primarily because of
multiple uses of a function with different arguments.
Different solutions to this conflict lead to
different representations of the storage, some of which are
examined by Accattoli and Barras~\cite{AccattoliB17} from the
perspective of time-cost analysis.
We recall a few solutions below that seem relevant to our token-guided
graph-rewriting.

One solution is to allow at most one assignment to each variable. This
is typically achieved by renaming bound variables during execution,
possibly symbolically. Examples for call-by-need evaluation are
Sestoft's abstract machines~\cite{Sestoft97}, and the storeless and
store-based abstract machines studied by Danvy and
Zerny~\cite{DanvyZ13}.
Our graph-rewriting abstract machine gives another example, as shown
by the simulation of the sub-machine semantics that resembles the
storeless abstract machine mentioned above. Variable renaming is
trivial in our machine, thanks to the use of graphs in which variables
are represented by mere edges.

Another solution is to allow multiple assignments to a variable,
with restricted visibility. The common approach is to pair a sub-term
with its own ``environment'' that maps its free variables to their
assigned computations, forming a so-called ``closure''. Conflicting
assignments are distributed to distinct localised
environments. Examples include Cregut's lazy variant~\cite{Cregut07}
of Krivine's abstract machine for call-by-need evaluation, and
Landin's SECD machine~\cite{Landin64} for call-by-value evaluation.
Fern{\'{a}}ndez and Siafakas~\cite{FernandezS09} refine this approach
for call-by-name and call-by-value evaluations,
based on closed reduction~\cite{FernandezMS05}, which restricts
beta-reduction to closed function arguments.
This suggests that the approach with localised environments can be
modelled in our setting by implementing closed reduction. The
implementation would require an extension of rewrite transitions
and a different strategy to trigger them, namely to eliminate
auxiliary doors of a $\oc$-box.

Finally, Fern{\'{a}}ndez and Siafakas~\cite{FernandezS09} propose
another approach to multiple assignments, in which multiple
assignments are augmented with binary strings so that each occurrence
of a variable can only refer to one of them. This approach is inspired
by the token-passing GoI, namely a token-passing abstract machine for
call-by-value evaluation, designed by Fern{\'{a}}ndez and
Mackie~\cite{FernandezM02}.
The augmenting binary strings come from paths of trees of binary
contractions, which
are used by the token-passing machine to represent shared assignments.
In our graph-rewriting machine, trees of binary contractions are
replaced with single generalised contraction nodes of arbitrary arity,
to achieve time efficiency.
Therefore, the counterpart of the paths over binary contractions is
simply connections over single generalised contraction nodes.

To wrap up, we introduced the DGoIM, which can interleave
token-passing GoI with
graph rewriting, using the token-passing as a guide.
As a case study, we showed how the DGoIM with the rewrites-first
interleaving can time-efficiently implement three evaluations:
call-by-need, left-to-right call-by-value and right-to-left
call-by-value.
These evaluations have different control over caching
intermediate results.
The difference boils down to different routing of the token in the
DGoIM, which is achieved by simply switching graph representations
(namely, nodes modelling function application) of terms.

The idea of using the token as a guide of graph rewriting was also
proposed by Sinot~\cite{Sinot05,Sinot06} for interaction nets.
He shows how using a token can make the rewriting system implement the
call-by-name, call-by-need and call-by-value evaluation strategies.
Our development in this work can be seen as a realisation of the
rewriting system as an abstract machine, in particular with explicit
control over copying sub-graphs.

The token-guided graph rewriting is a flexible framework
with which we can implement various evaluation strategies of the
lambda-calculus and analyse execution cost.
Our focus in this work was primarily on time efficiency. This is to
complement existing work on operational semantics given by
token-passing GoI, which
usually achieves space efficiency, and also to confirm that
introduction of graph rewriting to the semantics does not bring in any
hidden inefficiencies.
We believe that further refinements, not only of the interleaving
strategies of token routing and graph reduction, but also of the
graph representation, can be formulated to serve particular
objectives in the space-time execution efficiency trade-off, such as
\emph{full lazy evaluation}, as hinted by Sinot~\cite{Sinot05}.

As a final remark, the flexibility of our framework also allows us to
handle the operational semantics of exotic language features,
especially data-flow features.
One such feature is to turn a parametrised data-flow network into an
ordinary function that takes parameters as an argument and returns
the network, which we model using the token-guided graph
rewriting~\cite{CheungDGMR18}.
This feature can assist a common programming idiom of machine learning
tasks, in which a data-flow network is constructed as a program, and
then modified at run-time by updating values of parameters embedded
into the network.

\section*{Acknowledgement}
\noindent
We are grateful to Ugo Dal Lago and anonymous reviewers for
encouraging and insightful comments on earlier versions of this
work.
We thank Steven W.\ T.\ Cheung for helping us implement the on-line visualiser.
The second author is grateful to Michele Pagani for stimulating
discussions in the very early stages of this work.



\bibliographystyle{alpha}
\bibliography{ref}
\end{document}